\documentclass[draftcls, 11pt, onecolumn]{IEEEtran}
\usepackage{amssymb, amsmath, graphicx, paralist,subfigure,algorithm2e}
\usepackage{amsbsy}%,color,multicol}
\usepackage{floatflt} % text flows around figures

\usepackage{amsmath}
\usepackage{amssymb}
\usepackage{times}
\usepackage{graphicx}
\usepackage{xspace}
\usepackage{paralist} % compact and in-paragraph* lists
\usepackage{setspace} % buggy somehow
\usepackage{xypic}
\xyoption{curve}
\usepackage{latexsym}
\usepackage{theorem}
\usepackage{ifthen}
{\theoremheaderfont{\it} \theorembodyfont{\rmfamily}
\newtheorem{lem}{Lemma}
\newtheorem{theorem}{Theorem}
\newtheorem{definition}{Definition}

\newtheorem{corollary}{Corollary}

\newtheorem{proposition}[theorem]{Proposition}

}

\def\ve{{\bf e}}

\def\U{\mathbf{\Upsilon}}
\def\ve{\varepsilon}
\def\util{u(\mathbf{B},\mathbf{P})}

\def\ln{{\rm ln}}

\def\U{\mathbf{\Upsilon}}
\def\ra{\rightarrow}
\begin{document}
\title{\LARGE Higher Dimensional Consensus: Learning in Large-Scale Networks}
\author{Usman~A.~Khan, Soummya Kar, and Jos\'e M.~F.~Moura\\
\thanks{This work was partially supported by NSF under grant~\#~CNS-0428404, by ONR under grant \#~MURI-N000140710747, and by an IBM Faculty Award.}
            Department of Electrical and Computer Engineering\\
            Carnegie Mellon University, 5000 Forbes Ave, Pittsburgh, PA 15213\\
            \{ukhan, moura\}@ece.cmu.edu, soummyak@andrew.cmu.edu\\
            Ph: (412)268-7103 Fax: (412)268-3890
}

\maketitle
\begin{abstract}
 The paper presents higher dimension consensus (HDC) for large-scale networks. HDC generalizes the well-known average-consensus algorithm. It divides the nodes of the large-scale network into anchors and sensors. Anchors are nodes whose states are fixed over the HDC iterations, whereas sensors are nodes that update their states as a linear combination of the neighboring states. Under appropriate conditions, we show that the sensors' states converge to a linear combination of the anchors' states. Through the concept of anchors, HDC captures in a unified framework several interesting network tasks, including distributed sensor localization, leader-follower, distributed Jacobi to solve linear systems of algebraic equations, and, of course, average-consensus.

 In many network applications, it is of interest to \emph{learn} the weights of the distributed linear algorithm so that the sensors converge to a desired state. We term this \emph{inverse} problem the HDC \emph{learning} problem. We pose learning in HDC as a constrained non-convex optimization problem, which we cast in the framework of multi-objective optimization~(MPO) and to which we apply Pareto optimality. We prove analytically relevant properties of the MOP solutions and of the Pareto front from which we derive the solution to learning in HDC. Finally, the paper shows how the MOP approach resolves interesting tradeoffs (speed of convergence versus quality of the final state) arising in learning in HDC in resource constrained networks.
\end{abstract}

\newpage

\section{Introduction}\label{intro}
This paper provides a unified framework, \emph{high-dimensional consensus}~(HDC), for the analysis and design of linear distributed algorithms for large-scale networks--including distributed Jacobi algorithm \cite{tsit_book}, average-consensus \cite{jadbabailinmorse03,boyd:04,olfati_rev,Nedic,Fagnani,Tuncer}, distributed sensor localization \cite{usman_loc2:08}, distributed matrix inversion \cite{usman_icassp:07}, or leader-follower algorithms \cite{usman_icassp:09,khankarmoura-sep08}. These applications arise in many resource constrained large-scale networks, e.g., sensor networks, teams of robotic platforms, but also in cyber-physical systems like the smart grid in electric power systems. We view these systems as a collection of nodes interacting over a sparse communication graph. The nodes have, in general, strict constraints on their communication and computation budget so that only local communication and low-order computation is feasible at each node.

Linear distributed algorithms for constrained large-scale networks are iterative in nature; the information is fused over the iterations of the algorithm across the sparse network. In our formulation of HDC, we view the large-scale network as a graph with edges connecting sparsely a collection of nodes; each node is described by a state. The nodes are partitioned in anchors and sensors.  Anchors do not update their state over the HDC iterations, while the sensors iteratively update their states by a linear, possibly convex, combination of their neighboring sensors' states. The weights of this linear combination are the parameters of the HDC. For example, in sensor localization~\cite{usman_loc2:08}, the state at each node is its current position estimate. Anchors may be nodes instrumented with a GPS unit, knowing its precise location and the remaining nodes are the sensors that don't know their location and for which HDC iteratively updates their  state, i.e., their location, in a distributed fashion. The weights of HDC are for this problem the barycentric coordinates of the sensors with respect to a group of neighboring nodes, see~\cite{usman_loc2:08}.

We consider two main issues in HDC.

 \textbf{Analysis: Forward Problem} Given the HDC weights or parameters and the sparse underlying connectivity graph determine \begin{inparaenum}[(i)] \item under what conditions does the HDC converge; \item to what state does the HDC converge;  and \item what is the convergence rate. The forward problem establishes the conditions for convergence, the convergence rate, and the convergent state of the network. \end{inparaenum}

 \textbf{Learning: Inverse Problem} Given the desired state to which HDC should converge and the sparsity graph \emph{learn} the HDC parameters so that indeed HDC does converge to that state. Due to the sparsity constraints, it may not be possible for HDC to converge exactly to the desired state. An interesting tradeoff that we pursue is between the speed of convergence and the quality of the limiting HDC converging state, given by some measure of the error between the final state and the desired state. Clearly, the learning problem is an inverse problem that we will formulate as the minimization of a utility function under constraints.

 A naive formulation of the learning problem is not feasible. Ours is in terms of a constrained non-convex optimization formulation that we solve by casting it in the context of a multi-objective optimization problem~(MOP), \cite{par_book}. We apply to this MOP Pareto optimization. To derive the optimal Pareto solution, we need to characterize the Pareto front (locus of Pareto optimal solutions.) Although usually it is hard to determine the Pareto front and requires extensive iterative procedures, we exploit the structure of our problem to prove smoothness, convexity, strict decreasing monotonicity, and differentiability properties of the Pareto front. With the help of these properties, we can actually derive an efficient procedure to generate Pareto-optimal solutions to the MOP, determine the Pareto front, and find the solution to the learning problem. This solution is found by a rather expressive geometric argument.

\subsection{Organization of the Paper}
We now describe the rest of the paper. Section~\ref{prelim} introduces notation and relevant definitions, whereas Section~\ref{gcn} provides the problem formulation. We discuss the forward problem (analysis of HDC) in Section~\ref{fp_hdc} and the inverse problem (learning in large-scale networks) in Sections~\ref{lp_liln}--\ref{mo_ufn}. Finally, Section~\ref{conc} concludes the paper.

\section{Preliminaries}\label{prelim}
This section introduces the notation used in the paper and reviews relevant concepts from spectral graph theory and multi-objective optimization.

\subsection{Spectral Graph Theory}
Consider a sensor network with~$N$ nodes. We partition this network into~$K$ anchors and~$M$ sensors, such that~$N=K+M$. As discussed in Section~\ref{intro}, the anchors are the nodes whose states are fixed, and the sensors are the nodes that update their states as a linear combination of the states of their neighboring nodes. Let~$\kappa=\{1,\ldots,K\}$ be the set of anchors and let~$\Omega=\{K+1,\ldots,N\}$ be the set of sensors. The set of all of the nodes is then denoted by~$\Theta=\kappa\cup\Omega$.

We model the network by a directed graph, $\mathcal{G}=(V,\mathbf{A})$, where, $V=\{1,\ldots,N\}$, denotes the set of nodes. The interconnections among the nodes are given by the adjacency matrix,~$\mathbf{A}=\{a_{lj}\}$, where
\begin{eqnarray}
a_{lj} = \left\{\begin{array}{cc}
1, & l\leftarrow j,\\
0, & \mbox{otherwise},
\end{array}\right.
\end{eqnarray}
and the notation~$l\leftarrow j$ implies that node~$j$ can send information to node~$l$. The neighborhood,~$\mathcal{K}(l)$, at node~$l$ is
\begin{equation}
\mathcal{K}(l) = \{j~|~a_{lj}=1\}.
\end{equation}
The classification of nodes into sensors and anchors naturally induces the partitioning of the neighborhood,~$\mathcal{K}(l)$, at each sensor,~$l$, i.e.,
\begin{align}
\mathcal{K}_\Omega(l) &= \mathcal{K}(l) \cap \Omega,&\mathcal{K}_\kappa(l) &=\mathcal{K}(l) \cap \kappa,&
\end{align}
where~$\mathcal{K}_\Omega(l)$ and~$\mathcal{K}_\kappa(l)$ are the set of sensors and the set of anchors in sensor~$l$'s neighborhood, respectively.

As a graph can be characterized by its adjacency matrix, to every matrix we can associate a graph. For a matrix,~$\mathbf{\Upsilon}=\{\upsilon_{lj}\}\in\mathbb{R}^{N\times N}$, we define its associated graph by~$\mathcal{G}^{\mathbf{\Upsilon}}=(V^{\mathbf{\Upsilon}},\mathbf{A}^{\mathbf{\Upsilon}}),$ where~$V^{\mathbf{\Upsilon}}=\{1,\ldots,N\}$ and~$\mathbf{A}^{\mathbf{\Upsilon}}=\{a^{\mathbf{\Upsilon}}_{lj}\}$ is given by
\begin{equation}
a^{\mathbf{\Upsilon}}_{lj} = \left\{\begin{array}{cc}
1, & \upsilon_{lj}\neq 0,\\
0, & \upsilon_{lj}= 0.
\end{array}\right.
\end{equation}

The convergence properties of distributed algorithms depend on spectral properties of associated matrices. In the following, we recall the definition of spectral radius. The spectral radius,~$\rho(\mathbf{P})$, of a matrix,~$\mathbf{P}\in\mathbb{R}^{M\times M}$, is defined as
\begin{equation}
\rho(\mathbf{P}) = \max_i|\lambda_{i(\mathbf{P})}|,
\end{equation}
where~$\lambda_{i(\mathbf{P})}$ denotes the~$i$th eigenvalue of~$\mathbf{P}$. We also have
\begin{eqnarray}\label{rPdef}
\rho(\mathbf{P}) &=& \lim_{q\rightarrow\infty}\|\mathbf{P}^q\|^{1/q},
\end{eqnarray}
where~$\|\cdot\|$ is any matrix induced norm.

\subsection{Multi-objective Optimization Problem (MOP): Pareto-Optimality}\label{mop_sec}
In this subsection, we recall facts on multi-objective optimization theory that we will use to develop the solutions of the learning problem. We consider the following constrained optimization problem. Let~$\{f_k(\mathbf{y})\}_{k=1,\ldots,n}$ be real-valued functions,
\begin{eqnarray}
f_k : \mathcal{X} &\rightarrow& \mathbb{R}, \qquad \forall k
\end{eqnarray}
on some topological space,~$\mathcal{X}$ (in this work,~$\mathcal{X}$ will always be a finite-dimensional vector space). The vector objective function $\mathbf{f}(\mathbf{y})$ is
\begin{eqnarray}
\mathbf{f}(\mathbf{y}) =
\left[
\begin{array}{c}
f_1(\mathbf{y})\\
\vdots\\
f_n(\mathbf{y})
\end{array}
\right].
\end{eqnarray}
Let~$\{v_{\overline{k}}\}$ be a family of~$\overline{n}$ real-valued functions on~$\mathcal{X}$ representing the inequality constraints and~$\{w_{\underline{k}}\}$ be a family of~$\underline{n}$ real-valued functions on~$\mathcal{X}$ representing the equality constraints. The feasible set of solutions,~$\mathcal{Y}$, is defined as
\begin{equation}
\mathcal{Y} = \{\mathbf{y}\in\mathcal{X}~|~v_{\overline{k}}(\mathbf{y})\leq V_{\overline{k}},~\forall \overline{k}, \mbox{ and }w_{\underline{k}}(\mathbf{y})= W_{\underline{k}},~\forall {\underline{k}}\},
\end{equation}
where~$V_{\overline{k}},W_{\overline{k}}\in\mathbb{R}$. The multi-objective optimization problem (MOP) is given by
\begin{eqnarray}\label{mop_eq_eq}
\min_{\mathbf{y}\in\mathcal{Y}} \mathbf{f}(\mathbf{y}).
\end{eqnarray}
Note that the inequality constraints,~$v_{\overline{k}}(\mathbf{y})$, and the equality constraints,~$w_{\underline{k}}(\mathbf{y})$, appear in the set of feasible solutions and, thus, are implicit in~\eqref{mop_eq_eq}.

In general, the MOP has non-inferior solutions, i.e., the MOP has a set of solutions none of which is inferior to the other. The solutions of the MOP are, thus, categorized as Pareto-optimal \cite{par_book}, defined in the following.
\begin{definition}\label{po_def}[Pareto optimal solutions] A solution,~$\mathbf{y}^\ast$, is said to be a Pareto optimal (or non-inferior) solution of a~MOP if there exists no other feasible~$\mathbf{y}$ (i.e., there is no~$\mathbf{y}\in\mathcal{Y}$) such that~$\mathbf{f}(\mathbf{y})\leq\mathbf{f}(\mathbf{y}^\ast)$, meaning that~$f_k(\mathbf{y})\leq f_k(\mathbf{y}^\ast),~\forall k$, with strict inequality for at least one~$k$.
\end{definition}

The general methods to solve the MOP, for example, include the weighting method, the Lagrangian method, and the~$\mathbf{\varepsilon}$-constraint method. These methods can be used to find the Pareto-optimal solutions of the MOP. In general, these approaches require extensive iterative procedures to establish Pareto-optimality of a solution, see~\cite{par_book} for details.

%The~$\underline{\varepsilon}$-constraint problem, denoted by~$P_k(\underline{\varepsilon})$, is the following:
%\begin{eqnarray}
%\min_{\mathbf{y}\in\mathcal{Y}}f_k(\mathbf{y})&&\\\label{in_co}
%\mbox{subject to}\qquad f_j(\mathbf{y})&\leq&\varepsilon_j, \qquad j=1\ldots,n,~~j\neq k,
%\end{eqnarray}
%where~$\underline{\varepsilon} = [\varepsilon_1,\ldots,\varepsilon_{k-1},\varepsilon_{k+1},\ldots,\varepsilon_n]^T$. For a given solution,~$\mathbf{y}^\ast$, we use the symbol~$P_k(\varepsilon^\ast)$ to represent the problem~$P_k(\varepsilon)$, where~$\varepsilon_j=\varepsilon_j^\ast=f_j(\mathbf{y}^\ast)$,~$\forall j\neq k$. In other words,~$P_k(\varepsilon^\ast)$ is the minimization problem whose solution meets the inequality ($\leq$) constraint in~\eqref{in_co} at the equality. For details, see \cite{par_book}.

\section{Problem Formulation}\label{gcn}
Consider a sensor network with~$N$ nodes communicating over a network described by a directed graph,~$\mathcal{G}=(V,\mathbf{A})$. Let~$\mathbf{u}_k\in\mathbb{R}^{1\times m}$ be the state associated to the~$k$th anchor, and let~$\mathbf{x}_l\in\mathbb{R}^{1\times m}$ be the state associated to the~$l$th sensor. We are interested in studying linear iterative algorithms of the form
\begin{eqnarray}\label{alg1}
\mathbf{u}_k(t+1) &=& \mathbf{u}_k(t) = \mathbf{u}_k(0),\qquad\qquad\qquad\qquad\qquad\qquad\qquad t\geq0,\:k\in\kappa,\\\label{alg2}
\mathbf{x}_l(t+1) &=& p_{ll}\mathbf{x}_l(t) + \sum_{j\in\mathcal{K}_\Omega(l)}p_{lj}\mathbf{x}_j(t) + \sum_{k\in\mathcal{K}_\kappa(l)}b_{lk}\mathbf{u}_k(0),\qquad t\geq0,\:l\in\Omega,
\end{eqnarray}
where:~$t$ is the discrete-time iteration index; and ~$b_{lj}$'s and~$p_{lk}$'s are the state updating coefficients. We assume that the updating coefficients are constant over the components of the~$m$-dimensional state,~$\mathbf{x}_l(t)$.  We term distributed linear iterative algorithms of the form~\eqref{alg1}--\eqref{alg2} as \emph{Higher Dimensional Consensus (HDC)} algorithms\footnote{As we will show later, the HDC algorithms contain the conventional average consensus algorithms, \cite{boyd:04,olfati_rev}, as a special case. The notion of higher dimensions is technical and will be studied later.}~\cite{khankarmoura-sep08,usman_icassp:09}.

For the purpose of analysis, we write the HDC~\eqref{alg1}--\eqref{alg2} in matrix form. Define
\begin{align}
\mathbf{U}(t)&=\left[\mathbf{u}_1^T(t),\ldots,\mathbf{u}_K^T(t)\right]^T \in \mathbb{R}^{K\times m},&\mathbf{X}(t)&=\left[\mathbf{x}_{K+1}^T(t),\ldots,\mathbf{x}_N^T(t)\right]^T\in \mathbb{R}^{M\times m},&\\
\mathbf{P}&=\{p_{lj}\}\in\mathbb{R}^{M\times M},&
\mathbf{B}&=\{b_{lk}\}\in\mathbb{R}^{M\times K}.&
\end{align}
With the above notation, we write~\eqref{alg1}--\eqref{alg2} concisely as
\begin{eqnarray}\label{alg_mat}
\left[
\begin{array}{c}
\mathbf{U}(t+1)\\
\mathbf{X}(t+1)
\end{array}
\right] &=&
\left[
\begin{array}{cc}
\mathbf{I} & \mathbf{0}\\
\mathbf{B} & \mathbf{P}
\end{array}
\right]
\left[
\begin{array}{c}
\mathbf{U}(t)\\
\mathbf{X}(t)
\end{array}
\right],\\\label{Ceq}
\triangleq \mathbf{C}(t+1) &=& \mathbf{\Upsilon}\mathbf{C}(t).
\end{eqnarray}
Note that the graph,~$\mathcal{G}(\U)$, associated to the~$N\times N$ iteration matrix,~$\U$, is a subgraph of the communication graph,~$\mathcal{G}$. In other words, the sparsity of~$\U$ is dictated by the sparsity of the underlying sensor network, i.e., a non-zero element,~$\upsilon_{lj}$, in~$\U$ implies that node~$j$ can send information to node~$l$ in the original graph~$\mathcal{G}$. In the iteration matrix,~$\U$: the~$M\times M$ lower right submatrix,~$\mathbf{P}$, collects the updating coefficients of the~$M$ sensors with respect to the~$M$ sensors; and the lower left submatrix,~$\mathbf{B}$, collects the updating coefficients of the~$M$ sensors with respect to the~$K$ anchors. From~\eqref{alg_mat}, the matrix form of the HDC in~\eqref{alg2} is
\begin{eqnarray}
\label{gen_mod}\mathbf{X}(t+1) &=& \mathbf{PX}(t) + \mathbf{BU}(0),\qquad~t\geq0.
\end{eqnarray}
In this paper, we study the following two problems that arise in the context of the HDC.

\textbf{Analysis: Forward problem} Given an~$N$-node sensor network with a communication graph,~$\mathcal{G}$, the matrices~$\mathbf{B}$, and~$\mathbf{P}$, and the network initial conditions,~$\mathbf{X}(0)$ and~$\mathbf{U}(0)$; what are the conditions under which the HDC converges? What is the convergence rate of the HDC? If the HDC converges, what is the limiting state of the network?

\textbf{Learning: Inverse problem} Given an~$N$-node sensor network with a communication graph,~$\mathcal{G}$, and an~$M\times K$ weight matrix,~$\mathbf{W}$, learn the matrices~$\mathbf{P}$ and~$\mathbf{B}$ in~\eqref{gen_mod} such that the HDC converges to
\begin{equation}\label{Weq}
\lim_{t\rightarrow\infty}\mathbf{X}(t+1) = \mathbf{WU}(0),
\end{equation}
for every~$\mathbf{U}(0)\in\mathbb{R}^{K\times m}$; if multiple solutions exist, we are interested in finding a solution that leads to fastest convergence.

\section{Forward Problem: Higher Dimensional Consensus}\label{fp_hdc}
As discussed in Section~\ref{gcn}, the HDC algorithm is implemented as~\eqref{alg1}--\eqref{alg2}, and its matrix representation is given by~\eqref{Ceq}. We divide the study of the HDC in the following two categories.
\begin{enumerate}[(A)]
\item No anchors: $\mathbf{B=0}$
\item Anchors: $\mathbf{B\neq0}$
\end{enumerate}
We analyze these two cases separately. In addition, we also provide, briefly, practical applications where each of them is relevant.

\subsection{No anchors: $\mathbf{B=0}$}
In this case, the HDC reduces to
\begin{eqnarray}\label{gen_mod1}
\mathbf{X}(t+1) &=& \mathbf{PX}(t)\nonumber,\\
&=&\mathbf{P}^{t+1}\mathbf{X}(0).
\end{eqnarray}
An important problem covered by this case is average-consensus. As well known, if
\begin{equation}
\rho(\mathbf{P})=1,
\end{equation}
with $\mathbf{1}^T$ and $\mathbf{1}$ being the left and right eigenvectors of~$\mathbf{P}$, respectively, then we have
\begin{equation}\label{ffkk}
\lim_{t\ra\infty} \mathbf{P}^{t+1} = \dfrac{\mathbf{11}^T}{M},
\end{equation}
under some minimal network connectivity assumptions, where~$\mathbf{1}$ is the~$M\times 1$ column vector of~$1$'s and~$M$ is the number of sensors. The sensors converge to the average of the initial sensors' states. The convergence rate is dictated by the second largest (in magnitude) eigenvalue of the matrix~$\mathbf{P}$. For more precise and general statements in this regard, see for instance,~\cite{boyd:04,olfati_rev}. Average-consensus, thus, is a special case of the HDC, when~$\mathbf{B=0}$ and~$\rho(\mathbf{P})=1$. This problem has been studied in great detail. Relevant references include \cite{Mesbahi-Anchor,tsp07-K-M,Jadbabai,tsp06-K-A-M,karmoura-randomtopologynoise,Kashyap,Huang}. The rest of this paper deals entirely with the case~$\rho(\mathbf{P})<1$ and the term HDC subsumes the~$\rho(\mathbf{P})<1$ case, unless explicitly noted. When~$\mathbf{B=0}$, the HDC (with~$\rho(\mathbf{P})<1$) leads to~$\mathbf{X}_{\infty}=\mathbf{0}$, which is not very interesting.

\subsection{Anchors: $\mathbf{B\neq0}$}
This extends the average-consensus to ``higher dimensions'' (as will be explained in Section~\ref{con_ss}.) Lemma~\ref{lem_conv} establishes: (i) the conditions under which the HDC converges; (ii) the limiting state of the network; and (iii) the rate of convergence of the HDC.
\begin{lem}\label{lem_conv}
Let~$\mathbf{B}\neq0$. If
\begin{equation}
\label{spec_l}
\rho(\mathbf{P})<1,
\end{equation}
then the limiting state of the sensors,
\begin{equation}\label{gen_mod_lim}
\mathbf{X}_\infty \triangleq \lim_{t\rightarrow\infty}\mathbf{X}(t+1) = \left(\mathbf{I-P}\right)^{-1}\mathbf{BU}(0),
\end{equation}
and the error, $\mathbf{E}(t)= \mathbf{X}(t) - \mathbf{X}_{\infty}$, decays exponentially to~$0$ with exponent~$\ln(\rho(\mathbf{P}))$, i.e.,
\begin{equation}\label{lem_conv_eq}
\limsup_{t\rightarrow\infty}\dfrac{1}{t}\ln\|\mathbf{E}(t)\|~\leq~\ln(\rho(\mathbf{P})).
\end{equation}
\end{lem}
\begin{proof}
From~\eqref{gen_mod}, we note that
\begin{eqnarray}\label{gen_ic}
\mathbf{X}(t+1) &=& \mathbf{P}^{t+1}\mathbf{X}(0) + \sum_{k=0}^t\mathbf{P}^k\mathbf{B U}(0),\\
\Rightarrow\mathbf{X}_\infty &=& \lim_{t\rightarrow\infty}\mathbf{P}^{t+1}\mathbf{X}(0) + \lim_{t\rightarrow\infty}\sum_{k=0}^t\mathbf{P}^k\mathbf{B U}(0),
\end{eqnarray}
and~\eqref{gen_mod_lim} follows from~\eqref{spec_l} and Lemma~\ref{lem1} in Appendix~\ref{app1}. The error, $\mathbf{E}(t)$, is given by
\begin{eqnarray}\nonumber
\mathbf{E}(t)&=& \mathbf{X}(t) - (\mathbf{I-P})^{-1}\mathbf{BU}(0),\\\nonumber
&=& \mathbf{P}^{t}\mathbf{X}(0) + \sum_{k=0}^{t-1}\mathbf{P}^k\mathbf{BU}(0) - \sum_{k=0}^\infty\mathbf{P}^k\mathbf{BU}(0),\\\nonumber
&=& \mathbf{P}^{t}\ \left[\mathbf{X}(0)-\sum_{k=0}^\infty\mathbf{P}^k\mathbf{BU}(0)\right].
\end{eqnarray}
To go from the second equation to the third, we recall~\eqref{spec_l} and use~\eqref{lem2_eq2} from Lemma~\ref{lem1} in Appendix~\ref{app1}. Let~$\mathbf{R} = \mathbf{X}(0)-\sum_{k=0}^\infty\mathbf{P}^k\mathbf{BU}(0)$. To establish the convergence rate of~$\|\mathbf{E}(t)\|$, we have
\begin{eqnarray}\nonumber
\dfrac{1}{t}\ln\|\mathbf{E}(t)\| &=& \dfrac{1}{t}\ln\|\mathbf{P}^t \mathbf{R}\|,\\\nonumber
&\leq& \dfrac{1}{t}\ln\left(\|\mathbf{P}^t\| \|\mathbf{R}\|\right),\\
&\leq& \ln\|\mathbf{P}^t\|^{1/t}+\dfrac{1}{t}\ln \|\mathbf{R}\|.
\end{eqnarray}
Now, letting~$t\rightarrow\infty$ on both sides, we get
\begin{eqnarray}
\limsup_{t\rightarrow\infty}\dfrac{1}{t}\ln\|\mathbf{E}(t)\| &\leq & \limsup_{t\rightarrow\infty}\ln\|\mathbf{P}^t\|^{1/t}+\limsup_{t\rightarrow\infty}\dfrac{1}{t}\ln \|\mathbf{R}\|,\\
&=&\ln\lim_{t\rightarrow\infty}\|\mathbf{P}^t\|^{1/t},\\
&=&\ln\left(\rho(\mathbf{P})\right).
\end{eqnarray}
and~\eqref{lem_conv_eq} follows. The interchange of~$\lim$ and~$\ln$ is permissible because of the continuity of~$\ln$ and the last step follows from~\eqref{rPdef}.
\end{proof}
The above lemma shows that we require~\eqref{spec_l} for the HDC to converge. The limiting state of the network,~$\mathbf{X}_{\infty}$, is given by~\eqref{gen_mod_lim} and the error norm,~$\|\mathbf{E}(t)\|$, decays exponentially to zero with exponent~$\ln\left(\rho(\mathbf{P})\right)$. We further note that the limit state of the sensors,~$\mathbf{X}_\infty$, is independent of the sensors' initial conditions, i.e., the algorithm forgets the sensors' initial conditions and converges to~\eqref{gen_mod_lim} for any~$\mathbf{X}(0)\in\mathbb{R}^{M\times m}$. It is also straightforward to show that if~$\rho(\mathbf{P})\geq1$, then the HDC algorithm~\eqref{gen_mod} diverges for all~$\mathbf{U}(0)\in\mathcal{N}(\mathbf{B})$, where~$\mathcal{N}(\mathbf{B})$ is the null space of~$\mathbf{B}$. Clearly, the case~$\mathbf{U}(0)\in\mathcal{N}(\mathbf{B})$ is not interesting as it leads to~$\mathbf{X}_{\infty}=\mathbf{0}$.

\subsection{Consensus subspace}\label{con_ss}
We now define the consensus subspace as follows.
\begin{definition}[Consensus subspace]
Given the matrices,~$\mathbf{B}\in\mathbb{R}^{M\times K}$ and~$\mathbf{P}\in\mathbb{R}^{M\times M}$, the consensus subspace,~$\Xi$, is defined as
\begin{equation}
\Xi = \{\mathbf{X}_\infty~|~\mathbf{X}_\infty = \left(\mathbf{I-P}\right)^{-1}\mathbf{BU}(0),~\mathbf{U}(0)\in\mathbb{R}^{K\times m},~\rho(\mathbf{P})<1\}.
\end{equation}
\end{definition}
The dimension of the consensus subspace,~$\Xi$, is established in the following theorem.
\begin{theorem}
If~$K < M$ and~$\rho(\mathbf{P})<1$, then the dimension of the consensus subspace, ~$\Xi$, is
\begin{equation}
\mbox{dim}(\Xi) = m\mbox{rank}(\mathbf{B}) \leq mK.
\end{equation}
\end{theorem}
\begin{proof}
The proof follows From Lemma~\ref{lem_conv} and Lemma~\ref{lem_rank} in Appendix~\ref{app1}.
\end{proof}
Now, we formally define the dimension of the HDC.
\begin{definition}[Dimension]
\label{def:dimension} The dimension of the HDC algorithm is the dimension of the consensus subspace,~$\Xi$, normalized by~$m$, i.e.,
\begin{equation}
\mbox{dim}(\mbox{HDC}) = \dfrac{\mbox{dim}(\Xi)}{m} = \mbox{rank}(\mathbf{B}).
\end{equation}
\end{definition}
This definition is natural because the~HDC is a decoupled algorithm, i.e., HDC corresponds to~$m$ parallel algorithms, one for each column of~$\mathbf{X}(t)$.
So, the number of columns, $m$, in~$\mathbf{X}(t)$ is factored out in the definition of~$\mbox{dim}(\mbox{HDC})$. Each column of~$\mathbf{X}(t)$ lies in a subspace that is spanned by exactly~$\mbox{rank}(\mathbf{B})$ basis vectors. The number of these basis vectors is upper bounded by the number of anchors, i.e., is at most~$K$.

\subsection{Practical Applications of the HDC}
Several interesting problems can be framed in the context of HDC. We briefly sketch them below, for details, see \cite{usman_icassp:07,khankarmoura-sep08,usman_loc2:08,usman_icassp:09}.
\begin{itemize}
\item {\it Leader-follower algorithm} \cite{usman_icassp:09}: When there is only one anchor,~$K=1$, the sensors' states converge to the anchor state. With multiple anchors~$(K>1)$, under appropriate conditions, the sensors' states may be made to converge to a desired, pre-specified linear combination of the anchors' states.
\item {\it Sensor localization in~$m$-dimensional Euclidean spaces,~$\mathbb{R}^m$}: In~\cite{usman_loc2:08}, we choose the elements of the matrices~$\mathbf{P}=\{p_{lj}\}$ and~$\mathbf{B}=\{b_{lj}\}$ so that the sensor states converge to their exact locations when only,~$K=m+1$, anchors know their exact locations, for example, if equipped with a GPS.
\item {\it Jacobi algorithm for solving linear system of equations}, \cite{khankarmoura-sep08}: Linear systems of equations arise naturally in sensor networks, for example, power flow equations in power systems monitored by sensors or time synchronization algorithms in sensor networks. With appropriate choice of the matrices~$\mathbf{B}$ and~$\mathbf{P}$, it can be shown that the HDC algorithm~\eqref{gen_mod} is a distributed implementation of the Jacobi algorithm to solve the linear system.
\item {\it Distributed banded matrix inversion}: Algorithm~\eqref{gen_mod} followed by a non-linear collapse operator is employed in \cite{usman_icassp:07} to solve a banded matrix inversion problem, when the submatrices in the band are distributed among several sensors. This distributed inversion algorithm leads to distributed Kalman filters in sensor networks \cite{usman_tsp:07} using Gauss-Markov approximations by noting that the inverse of a Gauss-Markov covariance matrix is banded.
\end{itemize}

\subsection{Robustness of the HDC}\label{robu}
Robustness is key in the context of HDC, when the information exchange is subject to communication noise, packet drops, and imprecise knowledge of system parameters. In the context of sensor localization, we propose a modification to HDC in \cite{usman_loc2:08} along the lines of the Robbins-Monro algorithm \cite{kush_book} where the iterations are performed with a decreasing step-size sequence that satisfies a persistence condition. i.e., the step-sizes converge to zero but not too fast (this condition is well studied in the stochastic approximation literature, \cite{kush_book,Nevelson}). With such step-sizes, we show almost sure convergence of the sensor localization algorithm to their exact locations under broad random phenomenon, see~\cite{usman_loc2:08} for details. This modification is easily extended to the general class of HDC algorithms.

\section{Inverse Problem: Learning in Large-Scale Networks}\label{lp_liln}
As we briefly mentioned before, the inverse problem learns the parameter matrices ($\mathbf{B}$ and~$\mathbf{P}$) of the HDC such that HDC converges to a desired pre-specified state~\eqref{Weq}. For convergence, we require the spectral radius constraint~\eqref{spec_l}, and the matrices,~$\mathbf{B}$ and~$\mathbf{P}$, to follow the underlying communication network,~$\mathcal{G}$. In general, due to the spectral norm constraint and the sparseness (network) constraints, equation~\eqref{Weq} may not be met with equality. So, it is natural to relax the learning problem. Using Lemma~\ref{lem_conv} and~\eqref{Weq}, we restate the learning problem as follows.

Consider~$\varepsilon\in[0,1)$. Given an~$N$-node sensor network with a communication graph,~$\mathcal{G}$, and an~$M\times K$ weight matrix,~$\mathbf{W}$, solve the optimization problem:
\begin{eqnarray}
\label{s1}\inf_{\mathbf{B,P}}\|\left(\mathbf{I-P}\right)^{-1}\mathbf{B} &-& \mathbf{W}\|,\\
\mbox{subject to:}\qquad\label{s2}\mbox{Spectral radius constraint,}\qquad\rho\left(\mathbf{P}\right)&\leq&\varepsilon,\\
\label{s3}\mbox{Sparsity constraint,}\qquad\mathcal{G}^{\mathbf{\Upsilon}} &\subseteq& \mathcal{G},
\end{eqnarray}
for some induced matrix norm~$\|\cdot\|$. By Lemma~\ref{lem_conv}, if~$\rho\left(\mathbf{P}\right)\leq\varepsilon$, the convergence is exponential with exponent less than or equal to~$\ln(\varepsilon)$. Thus, we ask, given a pre-specified convergence rate,~$\ve$, what is the minimum error between the converged estimates,~$\lim_{t\rightarrow\infty}\mathbf{X}(t)$, and the desired estimates,~$\mathbf{WU}(0)$. Formulating the problem in this way naturally lends itself to a trade-off between the performance and the convergence rate.

In some cases, it may happen that the learning problem has an exact solution in the sense that there exist~$\mathbf{B,P}$, satisfying~\eqref{s2} and~\eqref{s3} such that the objective in~\eqref{s1} is~$0$. In case of multiple such solutions, we seek the one which corresponds to the fastest convergence, i.e., which leads to the smallest value of~$\rho(\mathbf{P})$. We may still formulate a performance versus convergence rate trade-off, if faster convergence is desired.

The learning problem stated as such is, in general, practically infeasible to solve because both~\eqref{s1} and~\eqref{s2} are non-convex in~$\mathbf{P}$. We now develop a more tractable framework for the learning problem in the following.

\subsection{Revisiting the spectral radius constraint~\eqref{s2}}\label{secbb}
We work with a convex relaxation of the spectral radius constraint. Recall that the spectral radius can be expressed as~\eqref{rPdef}. However, direct use of~\eqref{rPdef} as a constraint is, in general, not computationally feasible. Hence, instead of using the spectral radius constraint~\eqref{s2} we use a matrix induced norm constraint by realizing that
\begin{equation}\label{rPlb}
\rho(\mathbf{P}) \leq \|\mathbf{P}\|,
\end{equation}
for any matrix induced norm. The induced norm constraint, thus, becomes
\begin{equation}\label{neq_eq}
\left\|\mathbf{P}\right\|\leq\varepsilon.
\end{equation}
Clearly,~\eqref{rPlb} implies that any upper bound on~$\|\mathbf{P}\|$ is also an upper bound on~$\rho(\mathbf{P})$.

\subsection{Revisiting the sparsity constraint~\eqref{s3}}\label{seccc}
In this subsection, we rewrite the sparsity constraint~\eqref{s3} as a linear constraint in the design parameters,~$\mathbf{B}$ and~$\mathbf{P}$. The sparsity constraint ensures that the structure of the underlying communication network,~$\mathcal{G}$, is not violated. To this aim, we introduce an auxiliary variable,~$\mathbf{F}$, defined as
\begin{eqnarray}\label{min_not1}
\mathbf{F} &\triangleq& [\mathbf{B}~|~\mathbf{P}]\in\mathbb{R}^{M\times N}.
\end{eqnarray}
This auxiliary variable,~$\mathbf{F}$, combines the matrices~$\mathbf{B}$ and~$\mathbf{P}$ as they correspond to the adjacency matrix,~$\mathbf{A}(\mathcal{G})$, of the given communication graph,~$\mathcal{G}$, see the comments after~\eqref{Ceq}.

To translate the sparsity constraint into linear constraints on~$\mathbf{F}$ (and, thus, on~$\mathbf{B}$ and~$\mathbf{P}$), we employ a two-step procedure: (i)~First, we identify the elements in the adjacency matrix,~$\mathbf{A}(\mathcal{G})$, that are zero; these elements correspond to the pairs of nodes in the network where we do not have a communication link. (ii)~We then force the elements of~$\mathbf{F} = [\mathbf{B}~|~\mathbf{P}]$ corresponding to zeros in the adjacency matrix,~$\mathbf{A}(\mathcal{G})$, to be zero. Mathematically, (i) and (ii) can be described as follows.

(i)~Let the lower~$M\times N$ submatrix of the~$N\times N$ adjacency matrix,~$\mathbf{A}=\{a_{lj}\}$ (this lower part corresponds to~$\mathbf{F} = [\mathbf{P}~|~\mathbf{B}]$ as can be noted from~\eqref{Ceq}), be denoted by~$\underline{\mathbf{A}}$, i.e,
\begin{equation}\label{Abar}
\underline{\mathbf{A}}=\{\underline{a}_{ij}\}=\{a_{lj}\},\qquad l=K+1,\ldots,N,~~j=1,\ldots N,~~i=1,\ldots,M.
\end{equation}
Let~$\chi$ contain all pairs~$(i,j)$ for which~$\underline{a}_{ij}=0$.

(ii)~Let~$\{\mathbf{e}_i\}_{i=1,\ldots,M}$ be a family of~$1\times M$ row-vectors such that~$\mathbf{e}_i$ has a~$1$ as the~$i$th element and zeros everywhere else. Similarly, let~$\{\mathbf{e}^j\}_{j=1,\ldots,N}$ be a family of~$N\times 1$, column-vectors such that~$\mathbf{e}^j$ has a~$1$ as the~$j$th element and zeros everywhere else. With this notation, the~$ij$-th element,~$f_{ij}$, of~$\mathbf{F}$ can be written as
\begin{equation}
f_{ij} = \mathbf{e}_i\mathbf{Fe}^j.
\end{equation}

The sparsity constraint~\eqref{s3} is explicitly given by
\begin{equation}\label{sp_F}
\mathbf{e}_i\mathbf{Fe}^j=0,\qquad\forall~(i,j)\in\chi.
\end{equation}

\subsection{Feasible solutions}\label{fs_sec}
Consider~$\ve\in[0,1)$. We now define a set of  matrices,~$\mathcal{F}_{\leq\ve}\subseteq\mathbb{R}^{M\times N}$, that follow both the induced norm constraint~\eqref{neq_eq} and the sparsity constraint~\eqref{sp_F} of the learning problem. The set of feasible solutions is given by
\begin{equation}
\mathcal{F}_{\leq\ve} = \{\mathbf{F}_{\leq\ve}=[\mathbf{B}~|~\mathbf{P}]~|~\mathbf{e}_i\mathbf{Fe}^j=0,~\forall~(i,j)\in\chi,\mbox{ and }\|\mathbf{FT}\|\leq\varepsilon\},
\end{equation}
where
\begin{eqnarray}\label{min_not3}
\mathbf{T} &\triangleq&
\begin{bmatrix}
\mathbf{0}_{K\times M}\\
\mathbf{I}_M
\end{bmatrix}\in\mathbb{R}^{N\times M}.
\end{eqnarray}
With the matrix~$\mathbf{T}$ defined as above, we note that
\begin{equation}
\mathbf{P} = \mathbf{FT}.
\end{equation}

\begin{lem}
The set of feasible solutions,~$\mathcal{F}_{\leq\ve}$, is convex.
\end{lem}
\begin{proof}
Let~$\mathbf{F}_1,\mathbf{F}_2 \in \mathcal{F}_{\leq\ve}$, then
\begin{align}
\mathbf{e}_i\mathbf{F}_1\mathbf{e}^j&=0,~\forall~(i,j)\in\chi,&
\mathbf{e}_i\mathbf{F}_2\mathbf{e}^j&=0,~\forall~(i,j)\in\chi.&
\end{align}
For any~$0\leq\mu\leq 1$, and~$\forall~(i,j)\in\chi$,
\begin{eqnarray}\nonumber
\mathbf{e}_i\left(\mu\mathbf{F}_1+(1-\mu)\mathbf{F}_2\right)\mathbf{e}^j=\mu\mathbf{e}_i\mathbf{F}_1\mathbf{e}^j+(1-\mu)\mu\mathbf{e}_i\mathbf{F}_2\mathbf{e}^j=0.
\end{eqnarray}
Similarly,
\begin{eqnarray}\nonumber
\|\left(\mu\mathbf{F}_1+(1-\mu)\mathbf{F}_2\right)\mathbf{T}\|\leq\mu\|\mathbf{F}_1\mathbf{T}\|+(1-\mu)\|\mathbf{F}_2\mathbf{T}\|\leq\mu\varepsilon+(1-\mu)\varepsilon=\ve.
\end{eqnarray}
The first inequality uses the triangle inequality for matrix induced norms and the second uses the fact that, for~$i=1,2$,~$\mathbf{F}_i\in\mathcal{F}$ and~$\|\mathbf{F}_i\mathbf{T}\|\leq\varepsilon$.

Thus,~$\mathbf{F}_1,\mathbf{F}_2\in\mathcal{F}_{\leq\ve}\Rightarrow\mu\mathbf{F}_1+(1-\mu)\mathbf{F}_2 \in\mathcal{F}_{\leq\ve}$. Hence,~$\mathcal{F}_{\leq\ve}$ is convex.
\end{proof}
Similarly, we note that the sets,~$\mathcal{F}_{<\ve}$ and~$\mathcal{F}_{<1}$, are also convex.

\subsection{Learning Problem: An upper bound on the objective}\label{secdd}
In this section, we simplify the objective function~\eqref{s1} and give a tractable upper bound. We have the following proposition.
\begin{proposition}
Under the norm constraint $
%\begin{equation}\label{lllll}
\|\mathbf{P}\|<1,
%\end{equation}
$ then
\begin{equation}
\left\|\left(\mathbf{I-P}\right)^{-1}\mathbf{B} - \mathbf{W}\right\|
\leq\dfrac{1}{1-\left\|\mathbf{P}\right\|}\left\|\mathbf{B+PW-W}\right\|
\end{equation}
\end{proposition}
\begin{proof}
We manipulate~\eqref{s1} to obtain successively.
\begin{eqnarray}\nonumber
\left\|\left(\mathbf{I-P}\right)^{-1}\mathbf{B} - \mathbf{W}\right\| &=& \left\|\left(\mathbf{I-P}\right)^{-1}\left(\mathbf{B} - \left(\mathbf{I-P}\right)\mathbf{W}\right)\right\|,\\\nonumber
&\leq& \left\|\left(\mathbf{I-P}\right)^{-1}\right\|\left\|\left(\mathbf{B} - \left(\mathbf{I-P}\right)\mathbf{W}\right)\right\|,\\\nonumber
&=& \left\|\sum_k\mathbf{P}^k\right\|\left\|\left(\mathbf{B} - \left(\mathbf{I-P}\right)\mathbf{W}\right)\right\|,\\\nonumber
&\leq& \sum_k\|\mathbf{P}\|^k\left\|\left(\mathbf{B} - \left(\mathbf{I-P}\right)\mathbf{W}\right)\right\|,\\\label{s1ub}
&\leq&\dfrac{1}{1-\left\|\mathbf{P}\right\|}\left\|\mathbf{B+PW-W}\right\|.
\end{eqnarray}
To go from the second equation to the third, we use~\eqref{lem2_eq2} from Lemma~\ref{lem1} in Appendix~\ref{app1}. Lemma~\ref{lem1} is applicable here since~\eqref{rPlb} and given the norm constraint $\|\mathbf{P}\|<1$
%~\eqref{lllll}
 imply~$\rho(\mathbf{P})<1$. The last step is the sum of a geometric series which converges given $\|\mathbf{P}\|<1$. %~\eqref{lllll}.
\end{proof}
We now define the utility function,~$u(\mathbf{B,P})$, that we minimize instead of minimizing~$\|\left(\mathbf{I-P}\right)^{-1}\mathbf{B} - \mathbf{W}\|$. This is valid because~$\util$ is an upper bound on~\eqref{s1} and hence minimizing the upper bound leads to a performance guarantee. The utility function is
\begin{equation}\label{ut_fn}
u(\mathbf{B,P}) = \dfrac{1}{1-\left\|\mathbf{P}\right\|}\left\|\mathbf{B+PW-W}\right\|.
\end{equation}
With the help of the previous development, we now formally present the {\it Learning Problem}.
\begin{center}
    \begin{tabular}{ | p{6.5in} |}
    \hline
 {\it Learning Problem:} Given~$\ve\in[0,1)$, an~$N$-node sensor network with a sparse communication graph, $\mathcal{G}$, and a possibly full~$M\times K$ weight matrix, $\mathbf{W}$, design the matrices~$\mathbf{B}$ and~$\mathbf{P}$ (in~\eqref{gen_mod}) that minimize~\eqref{ut_fn}, i.e., solve the optimization problem
\begin{eqnarray}\label{lp_eq}
\label{c1}\inf_{[\mathbf{B}~|~\mathbf{P}]\in\mathcal{F}_{\leq\varepsilon}}
&~&u(\mathbf{B,P}).
\end{eqnarray}\\
    \hline
    \end{tabular}
\end{center}
Note that the induced norm constraint~\eqref{neq_eq} and the sparsity constraint~\eqref{sp_F} are implicit in~\eqref{c1}, as they appear in the set of feasible solutions,~$\mathcal{F}_{\leq\ve}$. Furthermore, the optimization problem in~\eqref{lp_eq} is equivalent to the following problem.
\begin{eqnarray}
\label{lp_eq2}\inf_{[\mathbf{B}~|~\mathbf{P}]\in\mathcal{F}_{\leq\varepsilon}\cap\{\|\mathbf{B}\|\leq b\}}&~&u(\mathbf{B,P}),
\end{eqnarray}
where~$b>0$ is a sufficiently large number. Since~\eqref{lp_eq2} involves the infimum of a continuous function,~$\util$, over a compact set,~$\mathcal{F}_{\varepsilon}\cap\{\|\mathbf{B}\|\leq b\}$, the infimum is attainable and, hence, in the subsequent development, we replace the infimum in~\eqref{lp_eq} by a minimum.

We view the $\min\util$ as the minimization of its two factors,~$1/(1-\|\mathbf{P}\|)$ and~$\|\mathbf{B+PW-W}\|$. In general, we need~$\|\mathbf{P}\|\rightarrow 0$ to minimize the first factor,~$1/(1-\|\mathbf{P}\|)$, and~$\|\mathbf{P}\|\rightarrow 1$ to minimize the second factor,~$\|\mathbf{B+PW-W}\|$ (we explicitly prove this statement later.) Hence, these two objectives are conflicting. Since, the minimization of the {\it non-convex} utility function,~$u(\mathbf{B},\mathbf{P})$, contains minimizing two coupled {\it convex} objective functions,~$\|\mathbf{P}\|$ and~$\|\mathbf{B+PW-W}\|$, we formulate this minimization as a multi-objective optimization problem (MOP). In the MOP, we consider separately minimizing these two convex functions. We then couple the MOP solutions using the utility function.

\subsection{Solution to the Learning Problem: MOP formulation}\label{sol_mop}
To solve the {\it Learning Problem} for every~$\varepsilon\in[0,1)$, we cast it in the context of a multi-objective optimization problem (MOP). We start by a rigorous definition of the MOP and later consider its equivalence to the {\it Learning Problem}. In the MOP formulation, we treat~$\|\mathbf{B+PW-W}\|$ as the first objective function,~$f_1$, and ~$\|\mathbf{P}\|$ as the second objective function,~$f_2$. The objective vector,~$\mathbf{f}(\mathbf{B,P})$, is
\begin{eqnarray}\label{vec_mop}
\mathbf{f}(\mathbf{B,P}) \triangleq
\left[
\begin{array}{c}
f_1(\mathbf{B,P})\\
f_2(\mathbf{B,P})
\end{array}
\right]
=
\left[
\begin{array}{c}
\|\mathbf{B+PW-W}\|\\
\|\mathbf{P}\|
\end{array}
\right].
\end{eqnarray}
The multi-objective optimization problem (MOP) is given by
\begin{eqnarray}\label{mop_eq}
\min_{[\mathbf{B}~|~\mathbf{P}]\in\mathcal{F}_{\leq1}} \mathbf{f}(\mathbf{B,P}),
\end{eqnarray}
where\footnote{Although the {\it Learning Problem} is valid only when~$\|\mathbf{P}\|<1$, the MOP is defined at~$\|\mathbf{P}\|=1$. Hence, we consider~$\|\mathbf{P}\|\leq1$ when we seek the MOP solutions.}
\begin{equation}
\mathcal{F}_{\leq1} = \{\mathbf{F}=[\mathbf{B}~|~\mathbf{P}]~:~\mathbf{e}_i\mathbf{Fe}^j=0,~\forall~(i,j)\in\chi,\mbox{ and }\|\mathbf{FT}\|\leq1\}.
\end{equation}

Before providing one of the main results of this paper on the equivalence of MOP and the {\it Learning Problem}, we set the following notation. We define
\begin{equation}
\varepsilon_{\mbox{exact}} = \min\{\|\mathbf{P}\|~|~(\mathbf{I-P})^{-1}\mathbf{B=W}, [\mathbf{B}~|\mathbf{P}]\in\mathcal{F}_{<1}\},
\end{equation}
where the minimum of an empty set is taken to be~$+\infty$. In other words,~$\varepsilon_{\mbox{exact}}$ is the minimum value of~$f_2=\|\mathbf{P}\|$ at which we may achieve an exact solution\footnote{An exact solution is given by~$[\mathbf{B}~|~\mathbf{P}]\in\mathcal{F}$ such that~$(\mathbf{I-P})^{-1}\mathbf{B=W}$ or when the infimum in~\eqref{s1} is attainable and is~$0$.} of the {\it Learning Problem}. A necessary condition for the existence of an exact solution is studied in Appendix~\ref{nc1}. If the exact solution is infeasible ($\notin\mathcal{F}_{<1}$), then~$\varepsilon_{\mbox{exact}}=\min\{\varnothing\}$, which we defined to be~$+\infty$. We let
\begin{equation}
\mathcal{E} = [0,1) \cap [0,\varepsilon_{\mbox{exact}}].
\end{equation}
The {\it Learning Problem} is interesting if~$\varepsilon\in\mathcal{E}$. We now study the relationship between the MOP and the \emph{Learning Problem}~\eqref{lp_eq}. Recall the notion of Pareto-optimal solutions of an MOP as discussed in Section~\ref{mop_sec}. We have the following theorem.
\begin{theorem}\label{main1}
Let~$\mathbf{B}_\varepsilon,\mathbf{P}_\varepsilon$, be an optimal solution of the {\it Learning Problem}, where~$\varepsilon\in\mathcal{E}$. Then,~$\mathbf{B}_\varepsilon,\mathbf{P}_\varepsilon$ is a Pareto-optimal solution of the MOP~\eqref{mop_eq}.
\end{theorem}

The proof relies on analytical properties of the MOP (discussed in Section~\ref{des}) and is deferred until Section~\ref{proof_th}. We discuss here the consequences of Theorem~\ref{main1}. Theorem~\ref{main1} says that the optimal solutions to the {\it Learning Problem} can be obtained from the Pareto-optimal solutions of the MOP. In particular, it suffices to generate the Pareto front (collection of Pareto-optimal solutions of the MOP) for the MOP and seek the solutions to the {\it Learning Problem} from the Pareto front. The subsequent Section is devoted to constructing the Pareto front for the MOP and studying the properties of the Pareto front.
%
%
%, we provide a sketch to our approach in the following.
%
%\begin{enumerate}[(i)]
%\item In Section~\ref{des}, we solve the MOP as an~$\varepsilon$-constraint problem.
%
%\item In Section~\ref{po_sols}, we prove that the solutions of the MOP are Pareto-optimal.
%
%\item In Section~\ref{pf_props}, we create a Pareto front that consists of all Pareto-optimal solutions and explore some properties of the Pareto front in the context of our learning problem.
%
%\item In Section~\ref{mo_ufn}, we impose the minimization of the utility function,~$u(\mathbf{B},\mathbf{P})$, on the Pareto front to get a Pareto optimal solution that gives us the lowest cost of~$u(\mathbf{B},\mathbf{P})$.
%\end{enumerate}
%In terms of the optimization,~$\varepsilon_{\mbox{exact}}~$ is given by
%\begin{equation}\label{abc_abc}
%\varepsilon_{\mbox{exact}} = \mbox{argmin}_{[\mathbf{B}~|~\mathbf{P}]\in\mathcal{F}} P_2(0).
%\end{equation}

\section{Multi-objective Optimization: Pareto Front}\label{des}
We consider the MOP~\eqref{mop_eq} as an~$\varepsilon$-constraint problem, denoted by~$P_k(\varepsilon)$ \cite{par_book}. For a two-objective optimization,~$n=2$, we denote the~$\varepsilon$-constraint problem as~$P_1(\varepsilon_2)$ or~$P_2(\varepsilon_1)$, where~$P_1(\varepsilon_2)$ is given by\footnote{All the infima can be replaced by minima in a similar way as justified in Section~\ref{secdd}. Further note that, for technical convenience, we use~$\mathcal{F}_{\leq1}$ and not~$\mathcal{F}_{<1}$, which is permissible because the MOP objectives are defined for all values of~$\|\mathbf{P}\|$.}
\begin{eqnarray}\label{meq1}
\min_{[\mathbf{B}~|~\mathbf{P}]\in\mathcal{F}_{\leq1}} f_1(\mathbf{B,P}) \qquad\mbox{subject to}\qquad f_2(\mathbf{B,P}) \leq \varepsilon_2.
\end{eqnarray}
and~$P_2(\ve_1)$ is given by
\begin{eqnarray}\label{meq12}
\min_{[\mathbf{B}~|~\mathbf{P}]\in\mathcal{F}_{\leq1}} f_2(\mathbf{B,P}) \qquad\mbox{subject to}\qquad f_1(\mathbf{B,P}) \leq \varepsilon_1.
\end{eqnarray}
In both~$P_1(\varepsilon_2)$ and~$P_2(\varepsilon_1)$, we are minimizing a real-valued convex function, subject to a constraint on the real-valued convex function over a convex feasible set. Hence, either optimization can be solved using a convex program~\cite{Boyd-CVXBook}. We can now write~$\ve_{\mbox{exact}}$ in terms of~$P_2(\varepsilon_1)$ as
\begin{equation}
\ve_{\mbox{exact}} =
\left\{
\begin{array}{cc}
P_2(0), & \mbox{if there exists a solution to }P_2(0),\\
+\infty,&\mbox{otherwise}.
\end{array}
\right.
\end{equation}

Using~$P_1(\varepsilon_2)$, we find the Pareto-optimal set of solutions of the MOP. We explore this in Section~\ref{po_sols}. The collection of the values of the functions,~$f_1$ and~$f_2$, at the Pareto-optimal solutions forms the Pareto front (formally defined in Section~\ref{pf_props}). We explore properties of the Pareto front, in the context of our learning problem, in Section~\ref{pf_props}. These properties will be useful in addressing the minimization in~\eqref{lp_eq} for solving the {\it Learning Problem}.

\subsection{Pareto-Optimal Solutions}\label{po_sols}
In general, obtaining Pareto-optimal solutions requires iteratively solving~$\varepsilon$-constraint problems \cite{par_book}, but we will show that the optimization problem,~$P_1(\varepsilon_2)$, results directly into a Pareto-optimal solution. To do this, we provide Lemma~\ref{lem_def} and its Corollary~\ref{cor_dec} in the following. Based on these, we then state the Pareto-optimality of the solutions of~$P_1(\varepsilon_2)$ in Theorem~\ref{thm_po}.

\begin{lem}\label{lem_def}
Let
\begin{eqnarray}\label{meq112}
[\mathbf{B}_0~|~\mathbf{P}_0]&=&\mbox{argmin}_{[\mathbf{B}~|~\mathbf{P}]\in\mathcal{F}_{\leq1}} P_1(\varepsilon_0).
\end{eqnarray}
If~$\varepsilon_0\in\mathcal{E}$, then the minimum of the optimization,~$P_1(\varepsilon_0)$, is attained at~$\varepsilon_0$, i.e.,
\begin{equation}
f_2(\mathbf{B}_0,\mathbf{P}_0) = \varepsilon_0.
\end{equation}
\end{lem}
\begin{proof}
Let the minimum value of the objective,~$f_1$, be denoted by~$\delta_0$, i.e.,
\begin{equation}
\delta_0 = f_1(\mathbf{B}_0,\mathbf{P}_0).
\end{equation}
We prove this by contradiction. Assume, on the contrary, that~$\|\mathbf{P}_0\| = \varepsilon^\prime < \varepsilon_0$. Define
\begin{equation}\label{alp0}
\alpha_0\triangleq\dfrac{1-\varepsilon_0}{1-\varepsilon^\prime}.
\end{equation}
Since,~$\varepsilon^\prime < \varepsilon_0<1$, we have~$0<\alpha_0<1$. For~$\alpha_0\leq\alpha<1$, we define another pair,~$\mathbf{B}_1,\mathbf{P}_1$, as
\begin{align}\label{B1d}
\mathbf{B}_1 &\triangleq \alpha\mathbf{B}_0,&\mathbf{P}_1 &\triangleq (1-\alpha)\mathbf{I} + \alpha\mathbf{P}_0.&
\end{align}
Clearly, this choice is feasible, as it does not violate the sparsity constraints of the problem and further lies in the constraint of the optimization in~\eqref{meq112}, since
\begin{eqnarray}
\|\mathbf{P}_1\| &\leq& (1-\alpha) + \alpha \varepsilon^\prime\leq 1 -\alpha(1-\varepsilon^\prime) \leq 1 - \alpha_0(1-\varepsilon^\prime) = \varepsilon_0.
\end{eqnarray}
With the matrices~$\mathbf{B}_1,\mathbf{P}_1$ in~\eqref{B1d}, we have the following value, $\delta_1$, of the objective function,~$f_1$,
\begin{eqnarray}\nonumber
\delta_1 &=& \|\mathbf{B}_1+\mathbf{P}_1\mathbf{W}-\mathbf{W}\|,\\\nonumber
&=& \left\|\alpha\mathbf{B}_0 + \left(\left(1-\alpha\right)\mathbf{I} + \alpha\mathbf{P}_0\right)\mathbf{W} - \mathbf{W}\right\|,\\\nonumber
&=& \left\|\alpha\mathbf{B}_0 + \alpha\mathbf{P}_0\mathbf{W} - \alpha\mathbf{W}\right\|,\\
&=& \alpha f_1(\mathbf{B}_0,\mathbf{P}_0) = \alpha\delta_0.
\end{eqnarray}
Since,~$\alpha<1$ and non-negative, we have~$\delta_1<\delta_0$. This shows that the new pair,~$\mathbf{B}_1,\mathbf{P}_1$, constructed from the pair,~$\mathbf{B}_0,\mathbf{P}_0$, results in a lower value of the objective function. Hence, the pair,~$\mathbf{B}_0,\mathbf{P}_0$, with~$\|\mathbf{P}_0\| = \varepsilon^\prime < \varepsilon_0$ is not optimal, which is a contradiction. Hence,~$f_2(\mathbf{B}_0, \mathbf{P}_0) = \varepsilon_0,$.
\end{proof}

Lemma~\ref{lem_def} shows that if a pair of matrices,~$\mathbf{B}_0,\mathbf{P}_0$, solves the optimization problem~$P_1(\varepsilon_0)$ with~$\varepsilon_0\in\mathcal{E}$, then the pair of matrices,~$\mathbf{B}_0,\mathbf{P}_0$, meets the constraint on~$f_2$ with equality, i.e.,~$f_2(\mathbf{B}_0, \mathbf{P}_0) = \varepsilon_0$.
%Using our notation from Section~\ref{mop_sec}, we say that~$\mathbf{B}_0,\mathbf{P}_0$, solves the optimization problem~$P_1(\varepsilon^\ast)$.
The following corollary follows from Lemma~\ref{lem_def}.

\begin{corollary}\label{cor_dec}
Let~$\varepsilon_0\in\mathcal{E}$, and
\begin{eqnarray}
[\mathbf{B}_0~|~\mathbf{P}_0]&=&\mbox{argmin}_{[\mathbf{B}~|~\mathbf{P}]\in\mathcal{F}_{\leq1}} P_1(\varepsilon_0),\\
\delta_0 &=& f_1(\mathbf{B}_0,\mathbf{P}_0).
\end{eqnarray}
Then,
\begin{equation}
\delta_0 < \delta_\varepsilon,
\end{equation}
for any~$\varepsilon<\varepsilon_0$, where
\begin{eqnarray}
[\mathbf{B}_\varepsilon~|~\mathbf{P}_\varepsilon]&=&\mbox{argmin}_{[\mathbf{B}~|~\mathbf{P}]\in\mathcal{F}_{\leq1}} P_1(\varepsilon),\\
\delta_\varepsilon &=& f_1(\mathbf{B}_\varepsilon,\mathbf{P}_\varepsilon).
\end{eqnarray}
\end{corollary}
\begin{proof}
Clearly, from Lemma~\ref{lem_def} there does not exist any~$\varepsilon<\varepsilon_0$ that results in a lower value of the objective function,~$f_1$.
\end{proof}
The above lemma shows that the optimal value of~$f_1$ obtained by solving~$P_1(\varepsilon)$ is strictly greater than the optimal value of~$f_1$ obtained by solving~$P_1(\varepsilon_0)$ for any~$\varepsilon<\varepsilon_0$.

The following theorem now establishes the Pareto-optimality of the solutions of~$P_1(\varepsilon)$.

\begin{theorem}\label{thm_po}
If~$\varepsilon_0\in\mathcal{E}$, then the solution~$\mathbf{B}_0,\mathbf{P}_0$, of the optimization problem,~$P_1(\varepsilon_0)$, is Pareto optimal.
\end{theorem}
\begin{proof}
Since,~$\mathbf{B}_0,\mathbf{P}_0$ solves the optimization problem,~$P_1(\varepsilon_0)$, we have~$\|\mathbf{P}_0\|=\ve_0$, from Lemma~\ref{lem_def}. Assume, on the contrary that~$\mathbf{B}_0,\mathbf{P}_0$, are not Pareto-optimal. Then, by definition of Pareto-optimality, there exists a feasible~$\mathbf{B},\mathbf{P}$, with
\begin{eqnarray}\label{pp_1}
f_1(\mathbf{B},\mathbf{P})&\leq& f_1(\mathbf{B}_0,\mathbf{P}_0),\\\label{pp2}
f_2(\mathbf{B},\mathbf{P})&\leq& f_2(\mathbf{B}_0,\mathbf{P}_0),
\end{eqnarray}
with strict inequality in at least one of the above equations. Clearly, if~$f_2(\mathbf{B},\mathbf{P})<f_2(\mathbf{B}_0,\mathbf{P}_0)$, then~$\|\mathbf{P}\|<\ve_0$ and~$\mathbf{B},\mathbf{P}$, are feasible for~$P_1(\ve_0)$. By Corollary~\ref{cor_dec}, we have~$f_1(\mathbf{B},\mathbf{P})> f_1(\mathbf{B}_0,\mathbf{P}_0)$. Hence,~$f_2(\mathbf{B},\mathbf{P})<f_2(\mathbf{B}_0,\mathbf{P}_0)$ is not possible.

On the other hand, if~$f_1(\mathbf{B},\mathbf{P})<f_1(\mathbf{B}_0,\mathbf{P}_0)$ then we contradict the fact that~$\mathbf{B}_0,\mathbf{P}_0$, are optimal for~$P_1(\ve_0)$, since by~\eqref{pp2},~$\mathbf{B},\mathbf{P}$, is also feasible for~$P_1(\ve_0)$.

Thus, in either way, we have a contradiction and~$\mathbf{B}_0,\mathbf{P}_0$ are Pareto-optimal.
\end{proof}

\subsection{Properties of the Pareto Front}\label{pf_props}
In this section, we formally introduce the Pareto front and explore some of its properties in the context of the {\it Learning Problem}. The Pareto front and their properties are essential for the minimization of the utility function,~$u(\mathbf{B},\mathbf{P})$ over~$\mathcal{F}_{\leq\varepsilon}$~\eqref{lp_eq}, as introduced in Section~\ref{secdd}.

Let~$\overline{\mathcal{E}}$ denote the closure of~$\mathcal{E}$. The Pareto front is defined as follows.
\begin{definition}\label{p_f}[Pareto front] Consider~$\ve\in\overline{\mathcal{E}}$. Let~$\mathbf{B}_\ve$,~$\mathbf{P}_\ve$, be a solution of~$P_1(\ve)$ then\footnote{This follows from Lemma~\ref{lem_def}. Also, note that since~$\mathbf{B}_\ve$,~$\mathbf{P}_\ve$, is a solution of~$P_1(\ve)$,~$\mathbf{B}_\ve$,~$\mathbf{P}_\ve$, is Pareto optimal from Theorem~\ref{thm_po}.}~$\varepsilon=f_2(\mathbf{B}_\ve,\mathbf{P}_\ve)$. Denote by~$\delta=f_1(\mathbf{B}_\ve,\mathbf{P}_\ve)$. The collection of all such~$(\varepsilon,\delta)$ is defined as the Pareto front.
\end{definition}

For a given~$\varepsilon\in\overline{\mathcal{E}}$, define~$\delta(\varepsilon)$ to be the minimum of the objective function,~$f_1$, in~$P_1(\varepsilon)$. By Theorem~\ref{thm_po},~$(\varepsilon,\delta(\varepsilon))$ is a point on the Pareto front. We now view the Pareto front as a function,~$\delta: \overline{\mathcal{E}}\longmapsto\mathbb{R}_+$, which maps every~$\varepsilon\in\overline{\mathcal{E}}$ to the corresponding~$\delta(\varepsilon)$. In the following development, we use the Pareto front, as defined in Definition~\ref{p_f}, and the function,~$\delta$, interchangeably. The following lemmas establish properties of the Pareto front.
\begin{lem}\label{lem_sdpf}
The Pareto front is strictly decreasing.
\end{lem}
\begin{proof}
The proof follows from Corollary~\ref{cor_dec}.
\end{proof}
\begin{lem}\label{lem_con}
The Pareto front is convex, continuous, and, its left and right derivatives\footnote{At~$\varepsilon=0$, only the right derivative is defined and at~$\varepsilon=\sup{\overline{\mathcal{E}}}$, only the left derivative is defined.} exist at each point on the Pareto front. Also, when~$\ve_{\mbox{exact}}=+\infty$, we have
\begin{equation}
\delta(1) = \lim_{\ve\rightarrow1} \delta(\ve) = 0.
\end{equation}
\end{lem}
\begin{proof}
Let~$\varepsilon=f_2(\cdot)$ be the horizontal axis of the Pareto front, and let~$\delta(\varepsilon) = f_1(\cdot)$ be the vertical axis. By definition of the Pareto front, for each pair~$(\varepsilon,\delta(\varepsilon))$ on the Pareto front, there exists matrices~$\mathbf{B}_\varepsilon,\mathbf{P}_\varepsilon$ such that
\begin{eqnarray}
\|\mathbf{P}_\varepsilon\| = \varepsilon,\qquad\mbox{ and }\qquad\|\mathbf{B}_\varepsilon+\mathbf{P}_\varepsilon\mathbf{W-W}\|=\delta(\varepsilon).
\end{eqnarray}

Let~$(\varepsilon_1,\delta(\varepsilon_1))$ and~$\varepsilon_2,\delta(\varepsilon_2)$ be two points on the Pareto front, such that~$\varepsilon_1<\varepsilon_2$. Then, there exists~$\mathbf{B}_1,\mathbf{P}_1$, and~$\mathbf{B}_2,\mathbf{P}_2$, such that
\begin{eqnarray}
\|\mathbf{P}_1\| &=& \varepsilon_1,\qquad\mbox{ and }\qquad\|\mathbf{B}_1+\mathbf{P}_1\mathbf{W-W}\|=\delta(\varepsilon_1),\\
\|\mathbf{P}_2\| &=& \varepsilon_2,\qquad\mbox{ and }\qquad\|\mathbf{B}_2+\mathbf{P}_2\mathbf{W-W}\|=\delta(\varepsilon_2).
\end{eqnarray}
For some~$0\leq\mu\leq1$, define
\begin{eqnarray}
\mathbf{B}_3 &=& \mu\mathbf{B}_1+(1-\mu)\mathbf{B}_2,\\
\mathbf{P}_3 &=& \mu\mathbf{P}_1+(1-\mu)\mathbf{P}_2.
\end{eqnarray}
Clearly,~$[\mathbf{B}_3~|~\mathbf{P}_3]\in\mathcal{F}_{\leq1}$ as the sparsity constraint is not violated and
\begin{equation}\label{e3b}
\|\mathbf{P}_3\| \leq \mu\|\mathbf{P}_1\|+(1-\mu)\|\mathbf{P}_2\| < 1,
\end{equation}
since~$\|\mathbf{P}_1\|<1$ and~$\|\mathbf{P}_2\|<1$. Let
\begin{equation}
\varepsilon_3 = \|\mathbf{P}_3\|,
\end{equation}
and let
\begin{equation}
z(\varepsilon_3) = \|\mathbf{B}_3 + \mathbf{P}_3\mathbf{W} - \mathbf{W}\|.
\end{equation}
We have
\begin{eqnarray}\nonumber
z(\varepsilon_3) &=& \|\mu\mathbf{B}_1+(1-\mu)\mathbf{B}_2 + (\mu\mathbf{P}_1+(1-\mu)\mathbf{P}_2)\mathbf{W} - \mathbf{W}\|,\\\nonumber
&=&\|\mu\mathbf{B}_1+\mu\mathbf{P}_1\mathbf{W} - \mu\mathbf{W} + (1-\mu)\mathbf{B}_2 + (1-\mu)\mathbf{P}_2\mathbf{W} - (1-\mu)\mathbf{W}\|,\\\nonumber
&\leq&\mu\|\mathbf{B}_1+\mathbf{P}_1\mathbf{W} - \mathbf{W}\|+(1-\mu)\|\mathbf{B}_2 + \mathbf{P}_2\mathbf{W} - \mathbf{W}\|,\\
&=&\mu\delta(\varepsilon_1)+(1-\mu)\delta(\varepsilon_2).
\end{eqnarray}
Since~$(\varepsilon_3,z(\varepsilon_3))$ may not be Pareto-optimal, there exists a Pareto optimal point,~$(\varepsilon_3,\delta(\varepsilon_3))$, at~$\varepsilon_3$ (from Lemma~\ref{lem_def}) and we have
\begin{eqnarray}\nonumber
\delta(\varepsilon_3)&\leq& z(\varepsilon_3),\\\label{com2}
&\leq&\mu\delta(\varepsilon_1)+(1-\mu)\delta(\varepsilon_2).
\end{eqnarray}
From~\eqref{e3b}, we have
\begin{equation}
\varepsilon_3\leq \mu\varepsilon_1+(1-\mu)\varepsilon_2,
\end{equation}
and since the Pareto front is strictly decreasing (from Lemma~\ref{lem_sdpf}), we have
\begin{equation}\label{com1}
\delta(\mu\varepsilon_1+(1-\mu)\varepsilon_2) \leq \delta(\varepsilon_3).
\end{equation}
From~\eqref{com1} and~\eqref{com2}, we have
\begin{equation}
\delta(\mu\varepsilon_1+(1-\mu)\varepsilon_2) \leq \mu\delta(\varepsilon_1)+(1-\mu)\delta(\varepsilon_2),
\end{equation}
which establishes convexity of the Pareto front. Since, the Pareto front is convex, it is continuous, and it has left and right derivatives \cite{rockaf_book}.

Clearly,~$\delta(1) = \lim_{\ve\rightarrow1}\delta(\ve)$ by continuity of the Pareto front. By choosing~$\mathbf{P=I}$ and~$\mathbf{B=0}$, we have~$\delta(1) = 0$. Note that~$(1,0)$ lies on the Pareto front when~$\ve_{\mbox{exact}}=+\infty$. Indeed, for any~$\mathbf{B,P}$ satisfying the sparsity constraints, we simultaneously cannot have
\begin{eqnarray}
\|\mathbf{P}\|&\leq&1,\\
\mbox{or}\qquad\|\mathbf{B+PW-W}\| &\leq& 0,
\end{eqnarray}
with strict inequality in at least one of the above equations. Thus, the pair~$\mathbf{B=0},\mathbf{P=I}$ is Pareto-optimal leading to~$\delta(1)=0$.
\end{proof}

\subsection{Proof of Theorem~\ref{main1}}\label{proof_th}
With the Pareto-optimal solutions of MOP established in Section~\ref{po_sols} and the properties of the Pareto front in Section~\ref{pf_props}, we now prove Theorem~\ref{main1}.

\begin{proof}
We prove the theorem by contradiction. Let~$\|\mathbf{P}_\varepsilon\|=\varepsilon^\prime\leq\varepsilon$, and~$\delta^\prime=\|\mathbf{B}_\varepsilon+\mathbf{P}_\varepsilon\mathbf{W}-\mathbf{W}\|$. Assume, on the contrary, that~$\mathbf{B}_\varepsilon,\mathbf{P}_\varepsilon$ is not Pareto-optimal. From Lemma~\ref{lem_def}, there exists a Pareto-optimal solution~$\mathbf{B}^\ast,\mathbf{P}^\ast$, at~$\varepsilon^\prime$, such that
\begin{equation}
\|\mathbf{P}^\ast\|=\varepsilon^\prime, \mbox{ and }\delta(\varepsilon^\prime) = \|\mathbf{B}^\ast + \mathbf{P}^\ast\mathbf{W} - \mathbf{W}\|,
\end{equation}
with~$\delta(\varepsilon^\prime) < \delta^\prime$, since~$\mathbf{B}_\varepsilon,\mathbf{P}_\varepsilon$, is not Pareto-optimal. Since,~$\|\mathbf{P}_\varepsilon\|=\varepsilon^\prime\leq\varepsilon$, the Pareto-optimal solution,~$\mathbf{B}^\ast,\mathbf{P}^\ast$, is feasible for the {\it Learning Problem}. In this case, the utility function for the Pareto-optimal solution,~$\mathbf{B}^\ast,\mathbf{P}^\ast$, is
\begin{eqnarray}
u(\mathbf{B}^\ast,\mathbf{P}^\ast) & = & \dfrac{1}{1-\left\|\mathbf{P}^\ast\right\|}\left\|\mathbf{B}^\ast+ \mathbf{P}^\ast\mathbf{W}-\mathbf{W}\right\|,\\
&=& \dfrac{(\varepsilon^\prime)}{1-\varepsilon^\prime},\\
&<&\dfrac{\delta^\prime}{1-\varepsilon^\prime},\\
&=&u(\mathbf{B}_\varepsilon,\mathbf{P}_\varepsilon).
\end{eqnarray}
Hence,~$\mathbf{B}_\varepsilon,\mathbf{P}_\varepsilon$ is not an optimal solution of the {\it Learning Problem}, which is a contradiction. Hence,~$\mathbf{B}_\varepsilon,\mathbf{P}_\varepsilon$ is Pareto-optimal.
\end{proof}
The above theorem suggests that it suffices to find the optimal solutions of the {\it Learning Problem} from the set of Pareto-optimal solutions, i.e., the Pareto front. The next section addresses the minimization of the utility function,~$\util$, and formulates the performance-convergence rate trade-offs.

\section{Minimization of the utility function}\label{mo_ufn}
In this section, we develop the solution of the {\it Learning Problem} from the Pareto front. The solution of the {\it Learning Problem}~\eqref{lp_eq} lies on the Pareto front as already established in Theorem~\ref{main1}. Hence, it suffices to choose a Pareto-optimal solution from the Pareto front that minimizes~\eqref{lp_eq} under the given constraints. In the following, we study properties of the utility function.

\subsection{Properties of the utility function}
With the help of Theorem~\ref{main1}, we now restrict the utility function to the Pareto-optimal solutions\footnote{Note that when~$\ve_{\mbox{exact}}=+\infty$, the solution~$\mathbf{B=0},\mathbf{P=I}$ is Pareto-optimal, but the utility function is undefined here, although the MOP is well-defined. Hence, for the utility function, we consider only the Pareto-optimal solutions with~$\|\mathbf{P}\|$ in~$\mathcal{E}$.}. By Lemma~\ref{lem_def}, for every~$\ve\in\mathcal{E}$, there exists a Pareto-optimal solution,~$\mathbf{B}_\ve,\mathbf{P}_\ve$, with
\begin{equation}
\|\mathbf{P}_\ve\| = \ve,\qquad\mbox{and}\qquad\|\mathbf{B}_\ve+\mathbf{P}_\ve\mathbf{W-W}\|=\delta(\ve).
\end{equation}
Also, we note that, for any Pareto-optimal solution,~$\mathbf{B},\mathbf{P}$, the corresponding utility function,
\begin{equation}
\util=\dfrac{\|\mathbf{B}+\mathbf{P}\mathbf{W-W}\|}{1-\|\mathbf{P}\|} = \dfrac{\delta(\|\mathbf{P}\|)}{1-\|\mathbf{P}\|}.
\end{equation}
This permits us to redefine the utility function as,~$u^\ast:\mathcal{E}\longmapsto\mathbb{R}_+$, such that, for any Pareto-optimal solution,~$\mathbf{B},\mathbf{P}$,
\begin{equation}
\util=u^\ast(\|\mathbf{P}\|)
\end{equation}
We establish properties of~$u^\ast$, which enable determining the solutions of the {\it Learning Problem}.

\begin{lem}\label{u_noninc_lem}
The function~$u^\ast(\ve)$, for~$\ve\in\mathcal{E}$, is non-increasing, i.e., for~$\ve_1, \ve_2\in\mathcal{E}$ with~$\ve_1<\ve_2$, we have
\begin{equation}\label{u_noninc_eq}
u^\ast(\ve_2)\leq u^\ast(\ve_1).
\end{equation}
Hence,
\begin{equation}
\min_{[\mathbf{B}~|~\mathbf{P}]\in\mathcal{F}_{\leq\varepsilon}} \util = u^\ast(\ve).
\end{equation}
\end{lem}
\begin{proof}
Consider~$\ve_1, \ve_2\in\mathcal{E}$ such that~$\ve_1<\ve_2$, then~$\ve_2$ is a convex combination of~$\ve_1$ and~$1$, i.e., there exists a~$0<\mu<1$ such that \begin{equation}
\ve_2 = \mu\ve_1 + (1-\mu).
\end{equation}
From Lemma~\ref{lem_def}, there exist~$\delta(\ve_1)$ and~$\delta(\ve_2)$ on the Pareto front corresponding to~$\ve_1$ and~$\ve_2$, respectively. Since the Pareto front is convex (from Lemma~\ref{lem_con}), we have
\begin{equation}
\delta(\ve_2)\leq\mu\delta(\ve_1)+(1-\mu)\delta(1).
\end{equation}
Recall that~$\delta(1) = 0~$; we have
\begin{eqnarray}\nonumber
u^\ast(\ve_2) &=&
\dfrac{\delta(\ve_2)}{1-\ve_2},\\ \nonumber
&=& \dfrac{\delta(\ve_2)}{1-\mu\ve_1 - (1-\mu)},\\\nonumber
&=&\dfrac{\delta(\ve_2)}{\mu(1-\ve_1)},\\
&\leq&\dfrac{\mu\delta(\ve_1)}{\mu(1-\ve_1)},
\end{eqnarray}
and~\eqref{u_noninc_eq} follows.

We now have
\begin{eqnarray}\nonumber
\min_{[\mathbf{B}~|~\mathbf{P}]\in\mathcal{F}_{\varepsilon}}  \util &=&
\min_{[\mathbf{B}~|~\mathbf{P}]\in\mathcal{F}_{\varepsilon}\mbox{ and }(\mathbf{B,P}) \mbox{ is Pareto-optimal}} \util ,\\\nonumber
&=& \min_{\|\mathbf{P}\|\leq\ve\mbox{ and }(\mathbf{B,P}) \mbox{ is Pareto-optimal}} \util,\\\nonumber
&=&\min_{0\leq\ve^\prime\leq\ve} u^\ast(\ve^\prime),\\
&=& u^\ast(\ve).
\end{eqnarray}
The first step follows from Theorem~\ref{main1}. The second step is just a restatement since the sparsity constraints are included in the MOP. The third step follows from the definition of~$u^\ast$ and finally, we use the non-increasing property of~$u^\ast$ to get the last equation.
\end{proof}

We now study the cost of the utility function. From Lemma~\ref{u_noninc_lem}, we note that this cost is non-increasing as~$\ve$ increases. When~$\ve_{\mbox{exact}}<1$, this cost is~$0$. When~$\ve_{\mbox{exact}}=+\infty$, we may be able to decrease the cost as~$\ve\rightarrow1$. We now define the limiting cost.
\begin{definition}\label{cinfff}[Infimum cost] The infimum cost,~$c_{{\inf}}$, of the utility function is defined as
\begin{equation}
c_{\inf} \triangleq
\left\{
\begin{array}{cc}
\lim_{\ve\rightarrow1} u^\ast(\ve), & \mbox{if }\ve_{\mbox{exact}}=+\infty,\\
0,&\mbox{otherwise}.
\end{array}
\right.
\end{equation}
\end{definition}
Clearly, the cost does not increase as~$\ve\rightarrow1$ from Lemma~\ref{u_noninc_lem}. If~$\ve_{\mbox{exact}}=+\infty$, it is not possible for the utility function,~$u^\ast(\ve)$, to attain~$c_{\inf}$, since~$u^\ast(\ve)$ is undefined at~$\|\mathbf{P}\|=1$. We note that when~$\ve_{\mbox{exact}}=+\infty$, the utility function can have a value as close as desired to~$c_{\inf}$, but it cannot attain~$c_{\inf}$. The following lemma establishes the cost of the utility function,~$u^\ast(\ve)$, as~$\ve\rightarrow1$.
\begin{lem}
If~$\ve_{\mbox{exact}}=+\infty$, then the infimum cost,~$c_{\inf}$, is the negative of the left derivative,~$D^-(\delta(\ve))$, of the Pareto front evaluated at~$\ve=1$.
\end{lem}
\begin{proof}
Recall that~$\delta(1)=0$. Then~$c_{\inf}$ is given by
\begin{eqnarray}\nonumber
c_{\inf} &=&\lim_{\ve\rightarrow1}u^\ast(\ve),\\\nonumber
&=& \lim_{\ve\rightarrow1}\dfrac{\delta(\ve)}{1-\ve},\\\nonumber
&=& \lim_{\ve\rightarrow1}\dfrac{\delta(\ve)-\delta(1)}{1-\ve},\\
&=& -D^-(\delta(\ve)) |_{\ve=1}.
\end{eqnarray}
\end{proof}

\subsection{Graphical Representation of the Analytical Results}
In this section, we graphically view the analytical results developed earlier. To this aim, we establish a graphical procedure using the following lemma.
\begin{lem}\label{lem_gr}
Let~$(\ve,\delta(\ve))$ be a point on the Pareto front and~$g(\ve)$ a straight line that passes through~$(\ve,\delta(\ve))$ and~$(1,0)$. The cost associated to the Pareto-optimal solution(s) corresponding to~$(\ve,\delta(\ve))$ is both the (negative) slope and the intercept (on the vertical axis) of~$g(\ve)$.
\end{lem}
\begin{proof}
We define the straight line,~$g(\ve)$, as
\begin{equation}
g(\ve) = c_1\ve + c_2,
\end{equation}
where~$c_1$ is its slope and~$c_2$ is its intercept on the vertical axis. Since~$g(\ve)$ passes through~$(\ve,\delta(\ve))$ and~$(1,0)$, its slope,~$c_1$, is given by
\begin{eqnarray}
c_1 = \dfrac{\delta(\ve) - 0}{\ve-1}= -u^\ast(\ve).
\end{eqnarray}
Since~$g(\ve)$ passes through~$(1,0)$, at~$\ve=1$ we have
\begin{eqnarray}
c_2 = [g(\ve) - c_1\ve]_{\ve=1} = g(1) - c_1 = u^\ast(\ve).
\end{eqnarray}
\end{proof}
Figure~\ref{pfig1} illustrates Lemma~\ref{lem_gr}, graphically. Let~$(\varepsilon^\ast,\delta^\ast)$ be a point on the Pareto front. The cost,~$c^{\ast}$, of the utility function,~$u^\ast(\varepsilon^\ast)$, is the intercept of the straight line passing through~$(\varepsilon^\ast,\delta^\ast)$ and~$(1,0)$.
\begin{figure}
\centering
\subfigure[]
{
    \label{pfig1}
    \includegraphics[width=1.9in]{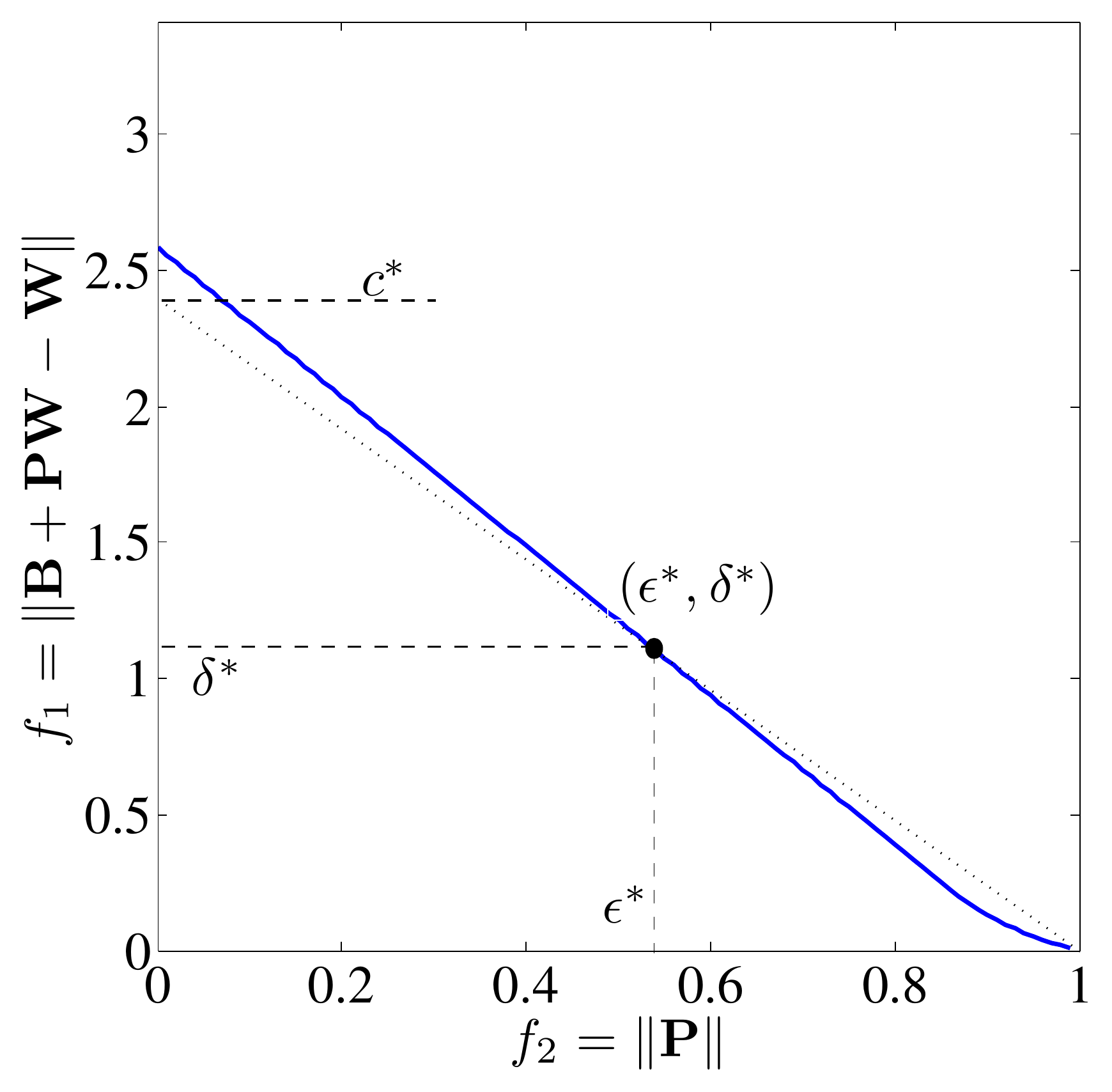}
}
\hspace{.1cm}
\subfigure[]
{
    \label{pfig2}
    \includegraphics[width=1.9in]{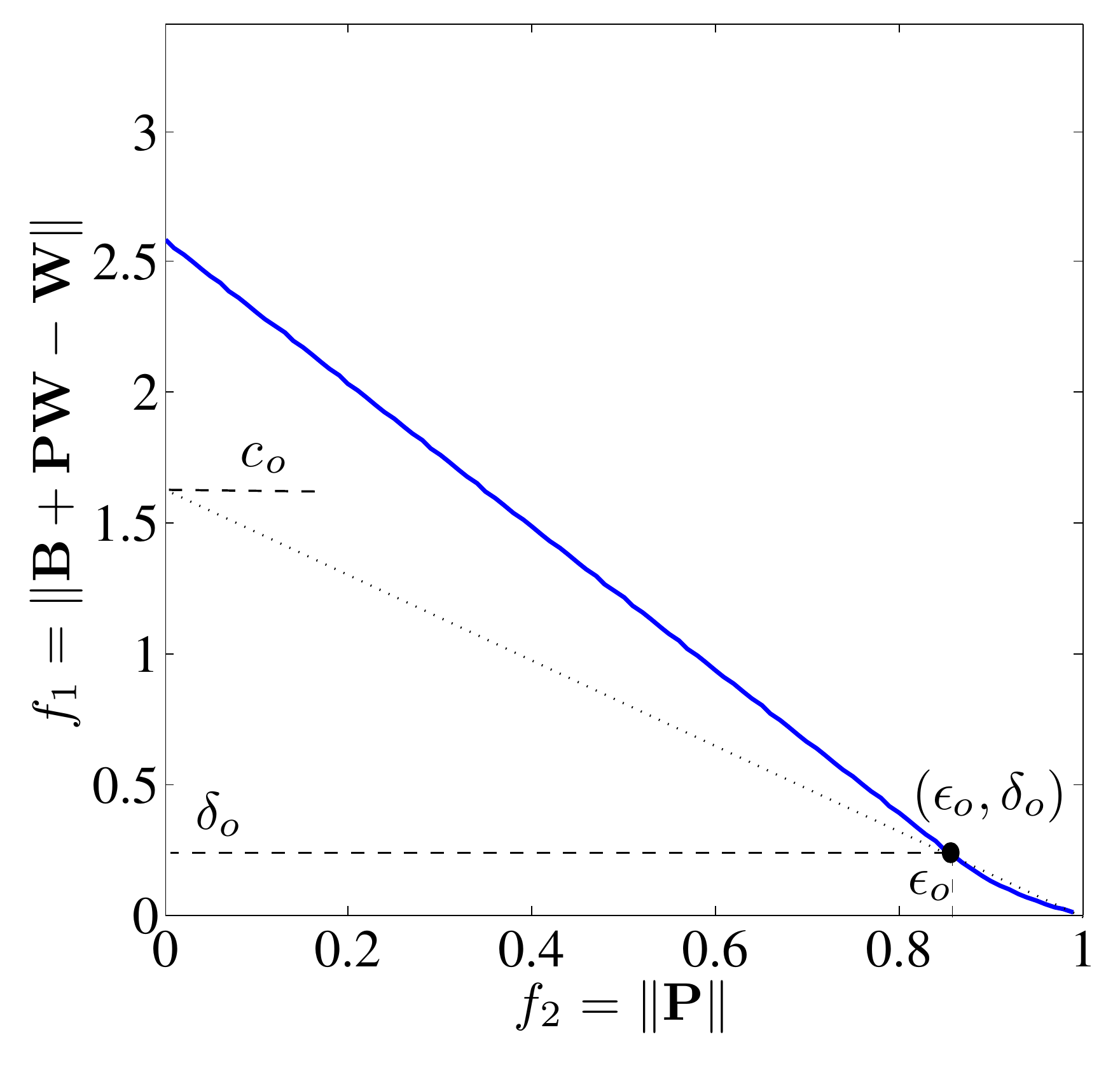}
}
\hspace{.1cm}
\subfigure[]
{
    \label{pfig3}
    \includegraphics[width=1.9in]{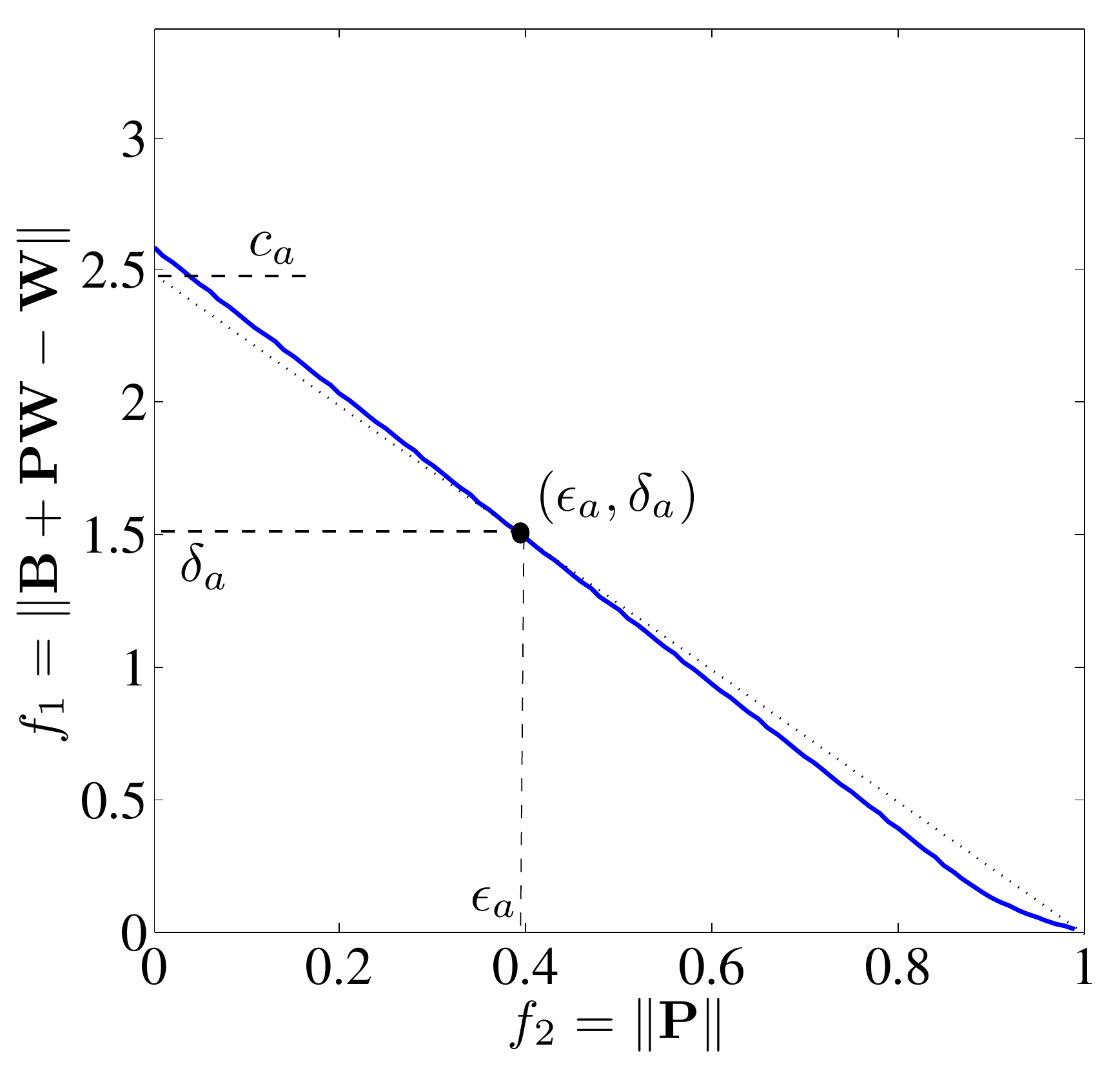}
}
\caption{(a) Graphical illustration of Lemma~\ref{lem_gr}. (b) Illustration of case (i) in performance-speed tradeoff. (c) Illustration of case (ii) in performance-speed tradeoff.}
\label{pfig}
\end{figure}

\subsection{Performance-Speed Tradeoff:~$\ve_{\mbox{exact}}=+\infty$}\label{infW}
In this case, no matter how large we choose~$\|\mathbf{P}\|$, the HDC does not converge to the exact solution. By Lemma~\ref{lem_conv}, the convergence rate of the HDC depends on~$\rho(\mathbf{P})$ and thus upper bounding~$\|\mathbf{P}\|$ leads to a guarantee on the convergence rate. Also, from Lemma~\ref{u_noninc_lem}, the utility function is non-increasing as we increase~$\|\mathbf{P}\|$. We formulate the {\it Learning Problem} as a performance-speed tradeoff. From the Pareto front and the constant cost straight lines, we can address the following two questions.

\begin{enumerate}[(i)]
\item Given a pre-specified performance,~$c_o$ (the cost of the utility function), choose a Pareto-optimal solution that results into the fastest convergence of the HDC to achieve~$c_o$.
    We carry out this procedure by drawing a straight line that passes the points~$(0,c_o)$
    and~$(1,0)$ in the Pareto plane. Then, we pick the Pareto-optimal solution from the Pareto front that lies on this straight line and also has the smallest value of~$\|\mathbf{P}\|$. See Figure~\ref{pfig2}.

\item Given a pre-specified convergence speed,~$\varepsilon_a$, of the HDC algorithm, choose a Pareto-optimal solution that results into the smallest cost of the utility function,~$\util$.
    We carry out this procedure by choosing the Pareto-optimal solution,~$(\varepsilon_a,\delta_a)$, from the Pareto front. The cost of the utility function for this solution is then the intercept (on the vertical axis) of the constant cost line that passes through both~$(\varepsilon_a,\delta_a)$ and~$(1,0)$. See Figure~\ref{pfig3}.
\end{enumerate}

We now characterize the steady state error. Let~$\mathbf{B}_o,\mathbf{P}_o$, be the operating point of the HDC obtained from either of the two tradeoff scenarios described above. Then, the steady state error in the limiting state,~$\mathbf{X}_{\infty}$, of the network when the HDC with~$\mathbf{B}_o,\mathbf{P}_o$ is implemented, is given by
\begin{equation}
e_{ss} = \|(\mathbf{I-P}_o)^{-1}\mathbf{B}_o-\mathbf{W}\|,
\end{equation}
which is clearly bounded above by~\eqref{ut_fn}.

\subsection{Exact Solution:~$\ve_{\mbox{exact}}<1$}\label{fW}
In this case, the optimal operating point of the HDC algorithm is the Pareto-optimal solution corresponding to~$(\varepsilon_{\mbox{exact}},0)$ on the Pareto front. A typical Pareto front in this case is shown in Figure~\ref{pfig4}, labeled as Case I. A special case is when the sparsity pattern of~$\mathbf{B}$ is the same as the sparsity of the weight matrix,~$\mathbf{W}$. We can then choose
\begin{equation}
\mathbf{B=W},\qquad\mathbf{P=0},
\end{equation}
as the solution to the {\it Learning Problem} and the Pareto front is a single point~$(0,0)$ shown as Case II in Figure~\ref{pfig4}.
\begin{figure}
\centering
\includegraphics[width=3in]{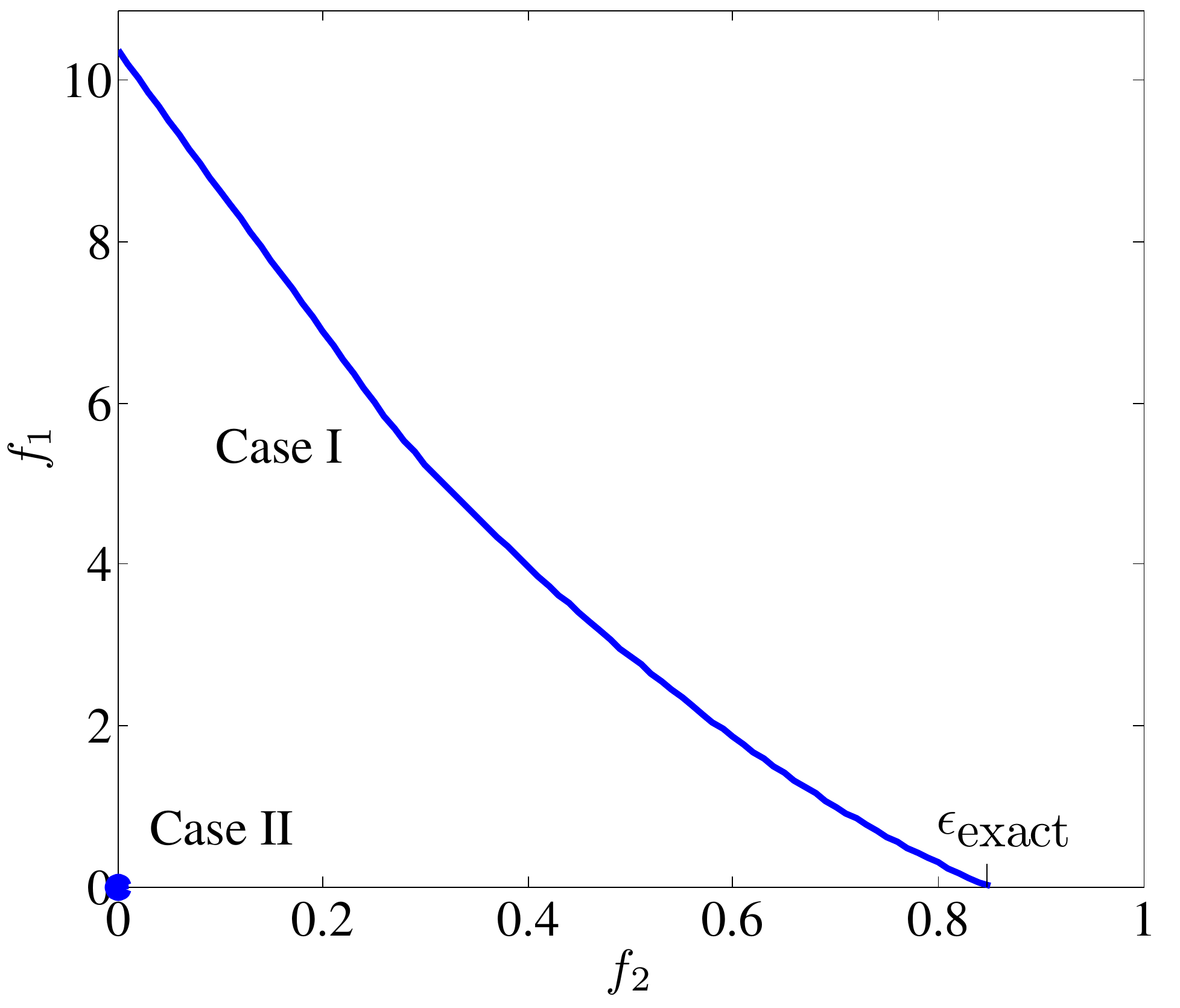}
\caption{Typical Pareto front.}
\label{pfig4}
\end{figure}

If it is desirable to operate the HDC algorithm at a faster speed than corresponding to~$\varepsilon_{\mbox{exact}}$, we can consider the performance-speed tradeoff in Section~\ref{infW} to get the appropriate operating point.

\section{Conclusions}\label{conc}
In this paper, we present a framework for the design and analysis of linear distributed algorithms. We present the analysis problem in the context of Higher Dimensional Consensus (HDC) algorithms that contains the average-consensus as a special case. We establish the convergence conditions, the convergent state and the convergence rate of the HDC. We also define the consensus subspace and derive its dimensions and relate them to the number of anchors in the network. We present the inverse problem of deriving the parameters of the HDC to converge to a given state as learning in large-scale networks. We show that the solution of this learning problem is a Pareto-optimal solution of a multi-objective optimization problem (MOP). We explicitly prove the Pareto-optimality of the MOP solutions. We then prove that the Pareto front (collections of the Pareto-optimal solutions) is convex and strictly decreasing. Using these properties of the MOP solutions, we solve the learning problem and also formulate performance-speed tradeoffs.
\small{
\appendices
\section{Important Results}\label{app1}
\begin{lem}\label{lem1}
If a matrix~$\mathbf{P}$ is such that \[\rho(\mathbf{P})<1,\] then
\begin{eqnarray}
\label{lem2_eq}
\lim_{t\rightarrow\infty}\mathbf{\mathbf{P}}^{t+1} &=&
\mathbf{0},
\\\label{lem2_eq2}
\lim_{t\rightarrow\infty}\sum_{k=0}^{t}
\mathbf{\mathbf{P}}^{k}&=&
\left(\mathbf{\mathbf{I}-\mathbf{P}}\right)^{-1}.
\end{eqnarray}
\end{lem}
\begin{proof}
The proof is straightforward.
\end{proof}

\begin{lem}\label{lem_rank}Let~$r_\mathbf{Q}$ be the rank of the~$M\times M$ matrix~$(\mathbf{I-P})^{-1}$, and~$r_\mathbf{B}$ the rank of the~$M\times n$ matrix~$\mathbf{B}$, then \begin{eqnarray}
\mbox{\rm rank}(\mathbf{I-P})^{-1}\mathbf{B} &\leq& \min(r_\mathbf{Q},r_\mathbf{B}),\\ \mbox{\rm rank}(\mathbf{I-P})^{-1}\mathbf{B} &\geq& r_\mathbf{Q} + r_\mathbf{B} - M.\end{eqnarray}\end{lem}
\begin{proof}The proof is available on pages~$95-96$ in~\cite{lin_alg_book_goog}.\end{proof}

\section{Necessary Condition}\label{nc1}
Below, we provide a necessary condition required for the existence of an exact solution of the {\it Learning Problem}.
\begin{lem}
Let~$\rho(\mathbf{P})<1$,~$K<M$, and let~$r_{\mathbf{W}}$ denote the rank of a matrix~$\mathbf{W}$. A necessary condition for $(\mathbf{I-P})^{-1}\mathbf{B=W}$ to hold is
\begin{equation}\label{rankw}
r_\mathbf{B} = r_\mathbf{W}.
\end{equation}
\end{lem}
\begin{proof}
Note that the matrix~$\mathbf{I-P}$ is invertible since~$\rho(\mathbf{P})<1$. Let~$\mathbf{Q=(I-P)}^{-1}$, then~$r_\mathbf{Q}=M$. From Lemma~\ref{lem_rank} in Appendix~\ref{app1} and since by hypothesis~$K<M$,
\begin{eqnarray}
\mbox{rank}(\mathbf{Q}\mathbf{B}) &\leq& r_\mathbf{B},\\
\mbox{rank}(\mathbf{Q}\mathbf{B}) &\geq& M + r_\mathbf{B} - M = r_\mathbf{B}.
\end{eqnarray}
The condition~\eqref{rankw} now follows, since from~\eqref{s1}, we also have
\begin{equation}
\mbox{rank}(\mathbf{Q}\mathbf{B})=r_\mathbf{W}.
\end{equation}
\end{proof}
}
\bibliographystyle{IEEEbib}
\bibliography{ref}

\end{document}